\documentclass[lettersize,journal]{IEEEtran}
\usepackage{amsmath,amsfonts,amsthm,amssymb,amsbsy}
\usepackage{algorithmic}
\usepackage{algorithm}
\usepackage{array}
\usepackage[caption=false,font=normalsize]{subfig}
\usepackage{textcomp}
\usepackage{stfloats}
\usepackage{url}
\usepackage{verbatim}
\usepackage{graphicx}
\usepackage{cite}
\usepackage{xcolor}
\usepackage{arydshln}

\newtheorem{assumption}{Assumption}
\newtheorem{remark}{Remark}
\newtheorem{theorem}{Theorem}
\newtheorem{lemma}{Lemma}

\begin{document}
\graphicspath{{./Figures/}}
\newcommand{\SAR} {\mathrm{SAR}}
\newcommand{\WBSAR} {\mathrm{SAR}_{\mathsf{WB}}}
\newcommand{\gSAR} {\mathrm{SAR}_{10\si{\gram}}}
\newcommand{\Sab} {S_{\mathsf{ab}}}
\newcommand{\Eavg} {E_{\mathsf{avg}}}
\newcommand{\ft}{f_{\textsf{th}}}
\newcommand{\alphatf}{\alpha_{24}}

\title{CRB-Based Resource Allocation in  Multi-User Uplink Transmissions
}
\author{
Xue Zhang, {\em Graduate Student Member, IEEE}, Abla Kammoun, {\em Member, IEEE}, \\ and Mohamed-Slim Alouini, {\em Fellow, IEEE}
\thanks{The authors are with Computer, Electrical and Mathematical Sciences
and Engineering (CEMSE) Division, Department of Electrical and Computer
Engineering, King Abdullah University of Science and Technology (KAUST), Thuwal 23955-6900, Saudi Arabia. (e-mail: xue.zhang@kaust.edu.sa; abla.kammoun@kaust.edu.sa; slim.alouini@kaust.edu.sa). 
}
\vspace{-8mm}
}
\maketitle

\vspace{-0.8cm}

\begin{abstract}
In this work, we study the design of receivers for uplink multi-user systems, aiming to estimate both the channel and the transmitted symbols. We consider two estimation strategies: (i) a joint estimation approach, where the channel and symbols are estimated simultaneously, and (ii) a sequential estimation approach, where the channel is first estimated and then used for symbol detection. For both strategies, we derive the Cram\'er-Rao Bound (CRB) for symbol estimation to characterize fundamental performance limits. When efficient receivers achieving the CRB exist, these bounds provide accurate lower bounds on the mutual information. In general, however, such receivers may not be available, and we instead use these same CRB-based metrics as practical proxies for achievable throughput. Leveraging tools from random matrix theory (RMT), we analyze the asymptotic behavior of these lower bounds under various asymptotic regimes for both estimation strategies. This analysis enables the derivation of generic power allocation guidelines that asymptotically maximize the proxy metrics. Simulation results confirm the accuracy of the asymptotic expressions and their effectiveness in guiding resource allocation decisions.
\end{abstract}


\section{Introduction}
\label{sec_introdution}
In multi-user systems, decoding the transmitted data symbols requires accurate knowledge of the channel. This makes the receiver inherently complex, as it must estimate the channel, the transmitted symbols, or both simultaneously. Classical approaches rely on training sequences (pilot symbols) to estimate channel state information (CSI)~\cite{hassibi2003much,biguesh2006training}, while blind estimation uses only the received data~\cite{noh2014new}. More advanced methods combine the two and perform semi-blind ~\cite{nayebi2017semi,al2021semi,lawal2023semi,huang2022semi,rekik2024fast} or joint channel-and-symbol estimation~\cite{aldana2003channel,srinivas2019iterative}. However, all these schemes involve solving high-dimensional and often non-convex estimation problems, leading to sophisticated receivers with significant computational complexity. As a result, analytically characterizing the performance of these estimation methods is challenging, and their impact on key metrics such as mutual information is not straightforward to quantify. This motivates the study of fundamental estimation limits, most notably the Cram\'er–Rao bound (CRB)~\cite{zhang2023joint,li2025rss,zhang2025ris}, which provides method-independent benchmarks on the accuracy achievable by any unbiased estimator. Depending on whether the transmitted data symbols are modeled as deterministic or stochastic, the authors in \cite{stoica1990performance,sandkuhler1987accuracy}  obtain deterministic or stochastic CRBs on the channel, respectively~. These bounds have been used to assess a wide range of estimation strategies: for example,~\cite{de1997cramer} compared the CRBs on the channel associated with blind, pilot-aided, and semi-blind channel estimation in single-input multiple-output (SIMO) systems, while~\cite{de1997cramer,de1997asymptotic} analyzed asymptotic regimes with long observation blocks. In massive multiple-input multiple-output (MIMO) systems,~\cite{nayebi2017semi} examined the asymptotic behavior of the CRB for channel estimation as the number of antennas grows, showing that it converges to limits corresponding to perfect data knowledge or fully orthogonal training. More recently, this analysis was refined in~\cite{zhang2025fundamental}, which characterized the CRB in a large-system regime where the number of antennas, users, and training symbols all grow simultaneously.

\par

A loss of mutual information between the channel input and output arises when the receiver has imperfect knowledge of the channel. Consequently, many initial works have investigated how channel estimation errors affect mutual information and capacity. A common modeling approach is to express the true channel as the sum of an estimate known to the receiver and a random estimation-error term. When substituted into the received signal model, the error term acts as additional interference caused by imperfect CSI, allowing one to study how its statistics affect mutual information. Based on this modeling, M\'edard derived in \cite{medard2000effect} the lower and upper bounds of information loss for Gaussian symbols and linked these bounds to the covariance matrix of the channel estimation error. Using the same modeling approach,  the work in~\cite{baltersee2001achievable} computed the achievable data rate of continuous-time varying flat fading interleaved MIMO channels. It was demonstrated that this rate is a function of the channel dynamics, determined by linear minimum mean square error (LMMSE) channel estimation and the corresponding error covariance matrix.
Under LMMSE channel estimation, Yoo and Goldsmith \cite{yoo2004capacity,medard2000effect} derived lower and upper bounds on the mutual information of frequency-flat MIMO channels with perfect feedback, assuming that the unknown data symbols and channel estimates are independent. In contrast, Hassibi and Hochwald \cite{hassibi2003much} derived an achievable rate for training-based MIMO transmission under multiplexed pilots, and showed that this rate depends explicitly on the MSE of the (optimal) LMMSE channel estimate. Their analysis further optimizes both pilot and data power allocations in order to maximize the resulting training-based achievable rate. 
Beyond these LMMSE-based analyses, a more abstract information-theoretic treatment was developed in \cite{lapidoth2002fading}. There, the impact of imperfect CSI is characterized without specifying any estimation model, and it is shown that when the receiver has only partial CSI, Gaussian codebooks and scaled nearest-neighbor detection can be strictly suboptimal. Practical receivers typically apply various post-processing operations to the channel output in order to form symbol estimates. However, by the data-processing inequality, such post-processing can only reduce the information available about the transmitted symbols. This observation motivates a careful investigation of how receiver design influences overall performance. This line of work was initiated in \cite{vosoughi2006effect}, which proposed a general mutual-information lower bound applicable to a wide range of receiver architectures, and established a connection between this bound and the Bayesian CRB (BCRB) associated with symbol estimation under superimposed training. 

\par

In the present work, we revisit this CRB-mutual-information framework but focus on multi-user systems and on the case where the pilots are time-multiplexed, as opposed to superimposed. Unlike \cite{vosoughi2006effect}, our objective is to derive power-allocation guidelines between the training and data transmissions. The motivation is that the performance analysis of practical receivers can be prohibitively intricate, rendering power optimization intractable. To overcome this difficulty, we compute the relevant CRB for symbol estimation and use it to construct mutual-information-related performance metrics. 
Although these CRB-based metrics do not constitute rigorous achievable rates, they provide analytically tractable surrogates that capture the impact of channel estimation accuracy on symbol detection performance. 
As such, they are particularly useful for revealing the fundamental tradeoffs between training and data transmission and for guiding resource allocation in scenarios where closed-form performance characterizations of practical receivers are unavailable. 

While both the classical Hassibi-Hochwald framework~\cite{hassibi2003much} and the present work investigate the allocation of resources between training and data transmission, the underlying design principles are fundamentally different. The Hassibi--Hochwald approach derives a lower bound on the achievable rate of a training-based communication scheme, where the channel is first estimated via pilot transmission and then used for data detection. The optimal training duration and power allocation are obtained by maximizing this capacity lower bound. In this sense, the performance metric is defined at the channel level and quantifies the information flow between transmitted symbols and received signals. In contrast, the proposed framework adopts an estimation-theoretic perspective and constructs CRB-based performance proxies that characterize the fundamental limits of symbol recovery accuracy. Rather than optimizing a training-based achievable rate, our approach evaluates how channel estimation quality impacts symbol detection performance and leverages this relationship to guide resource allocation. This shift from a channel-centric capacity formulation to a symbol-centric estimation framework enables a unified treatment of joint and sequential channel--symbol estimation strategies and facilitates tractable analysis in multi-user large-scale systems where closed-form training-based rate expressions are generally unavailable. Recently, dense uplink transmission and unsourced random access (URA) have also attracted significant attention due to their relevance to massive machine-type communications and large-scale connectivity scenarios~\cite{che2023massive,ozates2025unsourced}. Representative works include finite-blocklength analyses for MIMO URA systems, fluid antenna-assisted URA frameworks, and energy-efficient massive access designs. While the present work focuses on conventional multi-user uplink transmission with scheduled users and pilot-assisted communication, extending CRB-based large-system performance analysis and resource-allocation frameworks to dense massive access and URA scenarios constitutes an interesting direction for future research.

Following \cite{vosoughi2006effect}, we consider two symbol-estimation strategies. In strategy (i), the symbols and channels are estimated jointly, which encompasses semi-blind schemes that exploit both training and data observations. In strategy (ii), the channel is first estimated using the LMMSE channel estimator, after which symbol estimation is performed. For both strategies, we derive the corresponding CRBs, namely the deterministic CRB for strategy (i) and the BCRB for strategy (ii), and substitute these CRB expressions for the channel estimation error terms in the mutual-information lower bound. We then construct CRB-based proxies for the achievable throughput and derive random matrix theory (RMT) approximations of these proxies that depend only on the channel statistics. This enables us to formulate the corresponding power-allocation problems in a tractable manner. For strategy (i), the RMT expressions lead to closed-form approximations for the optimal power allocation. For strategy (ii), while closed-form formulas are not available, the resulting problem can be optimized efficiently through numerical methods. In our numerical evaluation, we compare the performance of the resulting power-allocation rules against a set of practical symbol estimators corresponding to both strategies (i) and (ii). The results show that the proposed guidelines remain highly accurate across all tested receivers. This demonstrates that CRB-based power allocation is particularly valuable when sophisticated symbol estimators are employed, especially in scenarios where closed-form performance characterizations are not available.

\par

\textit{Organization}: The remainder of this paper is organized as follows. Section~\ref{sec_model} introduces the system model. In Section~\ref{sec_JCSE}, we derive the deterministic CRB of the symbol estimator and analyze the asymptotic performance of the lower bound on the mutual information between the data symbols and their estimates for strategy (i). We perform a similar analysis for strategy (ii) in section \ref{sec:sequential} but with the BCRB instead. Simulation results are provided in Section~\ref{sec_sim}. Finally, Section~\ref{sec_con} concludes the article.

\par

\textit{Notations}: The vectors (matrices) are denoted by lower-case (upper-case) boldface characters. $\mathbf{I}_{M}$ denotes the $M\times M$ identity matrix. $\mathbb{E}[\cdot] $ represents the expectation operator. The trace  of $\mathbf{A}$ is  denoted by $\mathrm{tr}(\mathbf{A})$. Superscripts $T$, $H$, and $\ast$ denote the transpose, conjugate transpose, and complex conjugate, respectively.
The notation $[\mathbf{A}]_{i,j}$ refers to the $(i,j)$-th entry of the matrix $\mathbf{A}$, while $\mathrm{diag}(\mathbf{a})$ denotes the diagonal matrix whose diagonal entries are given by the elements of the vector $\mathbf{a}$.
The symbol $|\cdot|$ stands for the absolute value, and $\otimes$ denotes the Kronecker product.
For matrices, $\|\cdot\|_2$ denotes the spectral norm, whereas for vectors it denotes the Euclidean ($\ell_2$) norm.
The relation $\mathbf{A} \succeq \mathbf{B}$ means that $\mathbf{A}-\mathbf{B}$ is positive semidefinite.
The notation $\det(\mathbf{A})$ represents the determinant of $\mathbf{A}$.
The expression $x_n \xrightarrow{\mathrm{a.s.}} \overline{x}$ indicates almost-sure convergence of the sequence $\{x_n\}$ to $\overline{x}$.
The function $\delta(x)$ denotes the Dirac measure at point $x$. Finally, $\mathrm{blkdiag}(\cdot)$ denotes a block-diagonal matrix whose diagonal blocks are given by the matrices listed in the argument. For clarity, the main notations used throughout the paper are summarized in Table~\ref{table_symbol}.

\begin{table}[t]\label{table_symbol}
\centering
\caption{Summary of Main Symbols}
\small
\setlength{\tabcolsep}{3pt}
\renewcommand{\arraystretch}{1.2}
\begin{tabular}{l p{0.70\columnwidth}}
\noalign{\hrule height 1pt}
\textbf{Symbol} & \textbf{Definition} \\
\hline
$M$, $K$, $N$, $L$
& Numbers of antennas, users, symbols, and pilots \\
$c$, $\alpha$, $\beta$
& Asymptotic ratios $K/M$, $M/N$, and $L/N$ \\
$\mathbf{y}(n)$
& Received signal vector at time $n$ \\
$\mathbf{s}(n)$
& Transmitted symbol vector at time $n$ \\
$\mathbf{v}(n)$
& Noise vector at time $n$ \\
$\mathbf{g}_k$
& Channel vector of the $k$th user \\
$\mathbf{G}$
& Uplink channel matrix \\
$\mathbf{H}$, $\mathbf{B}$
& Small-scale and large-scale fading matrices \\
$\mathbf{S}_p$, $\mathbf{S}_d$
& Pilot and data symbol matrices \\
$\mathbf{Y}_p$, $\mathbf{Y}_d$
& Received pilot and data matrices \\
$\mathbf{V}_p$, $\mathbf{V}_d$
& Pilot and data noise matrices \\
$a$, $P$, $P_{\rm total}$
& Pilot, data, and total transmit powers \\
$x=a\beta$
& Average pilot power per symbol \\
$\sigma_v^2$
& Noise variance \\
$\hat{\mathbf{g}}_k$, $\tilde{\mathbf{g}}_k$
& LMMSE estimate and estimation error of user $k$ \\
$\mathcal{I}_{\rm joint}$, $\mathcal{I}_{\rm seq}$
& CRB-based mutual information approximations \\
$\overline{\mathcal{I}}_{\rm joint}$, $\overline{\mathcal{I}}_{\rm seq}$
& Asymptotic mutual information expressions \\
$\mathcal{T}_{\rm joint}$, $\mathcal{T}_{\rm seq}$
& Throughput approximations \\
\noalign{\hrule height 1pt}
\end{tabular}
\end{table}

\par

\section{Problem statement}
\label{sec_model}
\noindent{\bf System model.}
We consider a single-cell uplink system where $K$ users communicate with a base station (BS) equipped with $M$ antennas. Each user transmits a block of $N$ symbols, consisting of $L$ known training symbols followed by $N - L$ unknown data symbols. The received signal at the BS at time index $n \in \{0, \dots, N-1\}$ is given by
\begin{equation}\label{received_signal}
    \mathbf{y}(n) = \sum_{k=1}^{K} \mathbf{g}_k s_k(n) + \mathbf{v}(n),
\end{equation}
where $s_k(n) \in \mathbb{C}$ is the symbol transmitted by user $k$ at time $n$, $\mathbf{g}_k \in \mathbb{C}^{M}$ is the uplink channel vector from user $k$ to the BS, and $\mathbf{v}(n) \sim \mathcal{CN}(\mathbf{0}, \sigma_v^2 \mathbf{I}_M)$ is additive white Gaussian noise.
Define the transmitted symbol vector at time $n$ as $\mathbf{s}(n) = [s_1(n), \dots, s_K(n)]^T$, and collect the user channels in the matrix $\mathbf{G} = [\mathbf{g}_1, \dots, \mathbf{g}_K] \in \mathbb{C}^{M \times K}$.
Let $\mathbf{S}_p = [\mathbf{s}(0), \dots, \mathbf{s}(L-1)] \in \mathbb{C}^{K \times L}$ denote the matrix of known training symbols, and $\mathbf{S}_d = [\mathbf{s}(L), \dots, \mathbf{s}(N-1)] \in \mathbb{C}^{K \times (N - L)}$ denote the matrix of unknown data symbols. Similarly, define the received signal matrices during the training and data phases as ${\bf Y}_p=[{\bf y}(0),\cdots,{\bf y}(L)]$ and 
$\mathbf{Y}_d = [\mathbf{y}(L), \dots, \mathbf{y}(N-1)].$
Using this notation, the received signal model over the entire coherence block can be written compactly    
\begin{align}\mathbf{Y}_p &= \mathbf{G} \mathbf{S}_p + \mathbf{V}_p, \ \ \text{training period,} \\
    \mathbf{Y}_d &= \mathbf{G} \mathbf{S}_d + \mathbf{V}_d, \ \ \text{data transmission,}
\end{align}
where $\mathbf{V}_p$ and $\mathbf{V}_d$ are the corresponding noise matrices, with i.i.d. columns distributed as $\mathcal{CN}(\mathbf{0}, \sigma_v^2 \mathbf{I}_M)$. 
We assume that the training sequences are orthogonal and satisfy the normalization condition ${\bf X}_p:=\mathbf{S}_p \mathbf{S}_p^H = aL \mathbf{I}_K$, where $a > 0$ is a scaling parameter that reflects the average training power per user.

\noindent{\bf Objective.}
The goal of this work is to estimate the data symbol matrix $\mathbf{S}_d$ using the received signals collected during both the training and data transmission phases. To assess the fundamental limits of this estimation problem, we derive the CRB on the symbol estimation error, which serves as a benchmark for the best achievable performance. However, since it is not guaranteed that practical receivers can attain the CRB, we further derive a mutual information-related performance metric based on the CRB. Although this metric does not constitute a rigorous lower bound on the mutual information, it can serve as a useful proxy for system design. In particular, and as illustrated in Fig.~\ref{fig:receiver}, we investigate and compare two estimation strategies: (i) joint estimation of the channel and data symbols, and (ii) sequential estimation, where the channel is first estimated using the training data, and the resulting estimate is then used for symbol detection.

\begin{figure}
    \centering
    \includegraphics[width=1.0\linewidth]{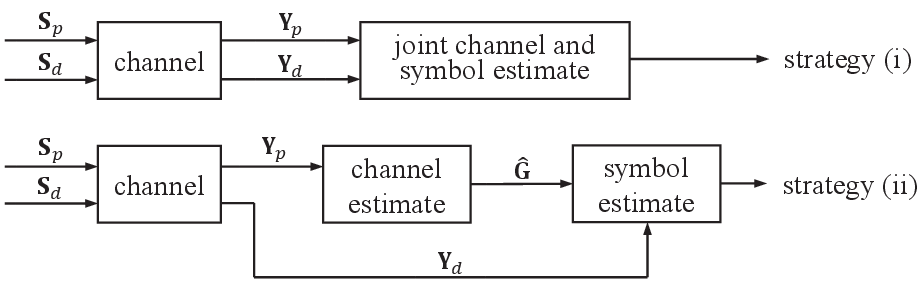}
    \caption{Receiver strategies for data-symbol recovery: (i) joint channel–symbol estimation and (ii) sequential channel and symbol estimation.}
    \label{fig:receiver}
\end{figure}

\par

\section{Joint Channel and Symbol 
Estimation}\label{sec_JCSE}
In this section, we consider Strategy (i), namely joint estimation of the channel and data symbols. This strategy is aligned with the principle of semi-blind channel estimation, as it exploits both the received training sequences and the received data symbols to infer the channel. We first derive a closed-form expression for the conventional deterministic CRB matrix associated with symbol estimation, since both the channel and the data symbols are treated as unknown but deterministic parameters in the joint estimation framework. The derivation is simplified by exploiting properties of the Kronecker product. Based on the resulting CRB expression, we then construct a mutual-information-related performance metric, which serves as a proxy for system capacity in the absence of CRB-achieving receivers. We further analyze the asymptotic behavior of this metric in the large-system regime and use the resulting expressions to develop generic resource allocation schemes.

\subsection{Symbol CRB Matrix}
We consider the joint estimation problem of the symbol matrix ${\bf S}_d$ and the channel ${\bf G}$ based on the received vector ${\bf Y}=[{\bf Y}_p, {\bf Y}_d]$ where both the unknown symbol matrix and the channel are treated as deterministic parameters. Let ${\bf g}_k$ denote the $k$-th column of matrix ${\bf G}$, and define $\pmb{\eta}=[\mathbf{s}^T(L),\ldots,\mathbf{s}^T(N-1),\mathbf{g}_1^T,\ldots,\mathbf{g}_K^T]^T$ the complex parameter vector which stacks the data symbols and the channel vectors to be estimated. Under the deterministic assumption for the unknown data matrix $\mathbf{S}_d$, each received vector $\mathbf{y}(n)$ follows an independent complex Gaussian distribution, i.e., $\mathbf{y}(n)\sim\mathcal{CN}(\mathbf{Gs}(n),\sigma_v^2\mathbf{I}_M)$ for $n=0, \ldots, N-1$. Accordingly, the likelihood function of the received matrix can be expressed as
\begin{align}
    p(\mathbf{Y};\pmb{\eta}) = \prod_{n=0}^{N-1} \frac{1}{(\pi \sigma_v^2)^M} 
    \exp\left(-\frac{1}{\sigma_v^2} \left\|\mathbf{y}(n)-\mathbf{G}\mathbf{s}(n)\right\|_2^2 \right). \notag
\end{align}
The corresponding log-likelihood function is given by
\begin{align}
    \ln p(\mathbf{Y}|\pmb{\eta}) 
    =& -MN \ln(\pi \sigma_v^2) 
    \notag\\
    &- \frac{1}{\sigma_v^2} \sum_{n=0}^{N-1} 
    \left\|\mathbf{y}(n)-\mathbf{G}\mathbf{s}(n)\right\|_2^2.
\end{align}
Let $\pmb{\eta}_s=[\mathbf{s}^T(L),\ldots,\mathbf{s}^T(N-1)]^{T}$ and $\pmb{\eta}_g=[\mathbf{g}_1^T,\ldots,\mathbf{g}_K^T]^{T}$, we define the complex Fisher information matrix as:
\begin{align}
    \pmb{\mathcal{J}}_{\eta\eta} =& \mathbb{E}\left[\left(\frac{\partial\ln{p(\mathbf{Y}|\pmb{\eta})}}{\partial\pmb{\eta}}\right)^H\frac{\partial\ln{p(\mathbf{Y}|\pmb{\eta})}}{\partial\pmb{\eta}}\right],
    \end{align}
    where the expectation is taken with respect to the distribution of the noise. We can write the Fisher information matrix in a block matrix form as:
    \begin{align}
\pmb{\mathcal{J}}_{\eta\eta}   =& \left[\begin{array}{cc}
        \pmb{\mathcal{J}}_{ss} & \pmb{\mathcal{J}}_{sg} \\
        \pmb{\mathcal{J}}_{sg}^H & \pmb{\mathcal{J}}_{gg} 
        \end{array}\right],
\end{align}
with
\begin{align}
\pmb{\mathcal{J}}_{ss}&=\mathbb{E}_{\mathbf{Y}|\pmb{\eta}}\left[\left(\frac{\partial\ln{p(\mathbf{Y}|\pmb{\eta})}}{\partial\pmb{\eta}_s}\right)^H\frac{\partial\ln{p(\mathbf{Y}|\pmb{\eta})}}{\partial\pmb{\eta}_s}\right]. \\
\pmb{\mathcal{J}}_{sg}&=\mathbb{E}_{\mathbf{Y}|\pmb{\eta}}\left[\left(\frac{\partial\ln{p(\mathbf{Y}|\pmb{\eta})}}{\partial\pmb{\eta}_s}\right)^H\frac{\partial\ln{p(\mathbf{Y}|\pmb{\eta})}}{\partial\pmb{\eta}_g}\right].\\
\pmb{\mathcal{J}}_{gg}&=\mathbb{E}_{\mathbf{Y}|\pmb{\eta}}\left[\left(\frac{\partial\ln{p(\mathbf{Y}|\pmb{\eta})}}{\partial\pmb{\eta}_g}\right)^H\frac{\partial\ln{p(\mathbf{Y}|\pmb{\eta})}}{\partial\pmb{\eta}_g}\right].
\end{align}
Here, the block $\pmb{\mathcal{J}}_{ss}$ describes the information about the data symbols when the channel is assumed to be known, and depends on the channel energy $\mathbf{G}^H\mathbf{G}$. The block $\pmb{\mathcal{J}}_{gg}$ represents the information available for channel estimation, which is provided by both pilot and data signals through the total covariance matrix ${\bf X}={\bf S}_p{\bf S}_p^{H}+{\bf S}_d{\bf S}_d^{H}$. The off-diagonal block $\pmb{\mathcal{J}}_{sg}$ reflects the coupling between symbol and channel estimation due to the uncertainty in the unknown data symbols. As a result, when the Fisher information matrix is inverted to obtain the CRB, the channel uncertainty is naturally separated into components associated with the subspace spanned by $\mathbf{G}$ and its orthogonal complement, which explains the appearance of the projection matrices.

Based on the above Fisher information matrix structure, the joint CRB for the symbol and channel parameters can be derived accordingly. The resulting matrix $\mathbf{CRB}(\mathbf{S}_d,\mathbf{G})$ admits a block structure, where the $MK \times MK$ lower-right block corresponds to the channel CRB, and the $(N-L)K \times (N-L)K$ upper-left block corresponds to the symbol CRB, denoted by $\mathbf{CRB}(\mathbf{S}_d)$. The channel CRB follows from the same analytical framework as in~\cite{zhang2025fundamental}. Since the focus of this work is on symbol-level performance characterization and its system-level implications, we concentrate on the symbol CRB component $\mathbf{CRB}(\mathbf{S}_d)$. By applying the standard block matrix inversion formula, $\mathbf{CRB}(\mathbf{S}_d)$ can be expressed as follows.

\begin{theorem}\label{theorem_complex_CRB}
The CRB for the symbol and the channel parameters is obtained as:
\begin{align}
        \mathbf{CRB}(\mathbf{S}_d,\mathbf{G}) :=& \left[\begin{array}{ll}\pmb{\mathcal{J}}_{ss} & \pmb{\mathcal{J}}_{sg}\\ \pmb{\mathcal{J}}_{sg}^{H}&\pmb{\mathcal{J}}_{gg}    
    \end{array}\right]^{-1},
    \end{align}
    where $\pmb{\mathcal{J}}_{ss}={\bf I}_{N-L}\otimes \frac{1}{\sigma_v^2}{\bf G}^{H}{\bf G}$, $\pmb{\mathcal{J}}_{gg}=\frac{1}{\sigma_v^2}{\bf X}^{\ast}\otimes {\bf I}_M$, and 
    \begin{align}
    \pmb{\mathcal{J}}_{sg}&=\begin{bmatrix}
    \frac{1}{\sigma_v^2}{\bf s}^{T}(L)\otimes {\bf G}^{H} \\ \vdots \\ \frac{1}{\sigma_v^2}{\bf s}^{T}(N-1)\otimes {\bf G}^{H}
    \end{bmatrix}.
    \end{align} Accordingly, the CRB for the symbol matrix $\mathbf{S}_d$ is given by
    \begin{align}\label{CRB_Symbol}
    {\bf CRB}({\bf S}_d):&=(\pmb{\mathcal{J}}_{ss}-\pmb{\mathcal{J}}_{sg}(\pmb{\mathcal{J}}_{gg})^{-1}\pmb{\mathcal{J}}_{sg}^{H})^{-1} \notag\\ &=\sigma_v^2({\bf I}_{N-L}-{\bf S}_d^{T}({\bf X}^\ast)^{-1}{\bf S}_d^{\ast})^{-1}\otimes ({\bf G}^{H}{\bf G})^{-1} \notag\\
    &=\sigma_v^2(\frac{1}{aL}{\bf S}_d^T{\bf S}_d^\ast+{\bf I}_{N-L})\otimes ({\bf G}^{H}{\bf G})^{-1}.
    \end{align}
    \begin{proof}
        See Appendix~\ref{proof_theorem_1}.
    \end{proof}
\end{theorem}

Let $\hat{\bf s}$ be an unbiased estimate of ${\bf s}:={\rm vec}({\bf S}_d)$. We assume that all symbols ${\bf s}$ are zero-mean independent complex Gaussian variables with $\mathbb{E}\{|s(n)|^2\}=P$, for all $n=L,\cdots,N-1$.
Let $\tilde{\bf s}=\hat{\bf s}-{\bf s}$ and ${\bf R}_{\tilde{\bf s}}=\mathbb{E}[\tilde{\bf s}\tilde{\bf s}^{H}]$. Then,
$$
{\bf R}_{\tilde{\bf s}}\succeq \mathbb{E}_{{\bf s}}[{\bf CRB}({\bf S}_d)]=\sigma_v^{2}\Big(\frac{KP}{aL}+1\Big){\bf I}_{N-L}\otimes ({\bf G}^{H}{\bf G})^{-1}.
$$

\par

\subsection{Mutual Information for Gaussian Symbols}
In this section, we follow the approach in ~\cite{vosoughi2006effect} to relate the error covariance of any unbiased symbol estimator to a lower bound on the mutual information. Let $\hat{\mathbf{s}}$ be any unbiased estimator of the transmitted symbol vector $\mathbf{s}$, and define the estimation error as $\tilde{\mathbf{s}} = \hat{\mathbf{s}} - \mathbf{s}$. The mutual information between $\mathbf{s}$ and $\hat{\mathbf{s}}$ satisfies:
\begin{align}
I(\mathbf{s}; \hat{\mathbf{s}}) &= \mathcal{H}(\mathbf{s}) - \mathcal{H}(\mathbf{s}|\hat{\mathbf{s}}) \notag\\
&= \mathcal{H}(\mathbf{s}) - \mathcal{H}(\mathbf{s} - \hat{\mathbf{s}} \mid \hat{\mathbf{s}}) \notag\\
&= \mathcal{H}(\mathbf{s}) - \mathcal{H}(\tilde{\mathbf{s}} \mid \hat{\mathbf{s}}) \notag\\
&\geq \mathcal{H}(\mathbf{s}) - \mathcal{H}(\tilde{\mathbf{s}}) \label{eq:used} \\
&\geq \mathcal{H}(\mathbf{s}) - \log_2\left( \det(\pi e\, \mathbf{R}_{\tilde{\mathbf{s}}}) \right). \label{eq:last}
\end{align}
Inequality~\eqref{eq:used} follows from the fact that conditioning reduces differential entropy, and inequality~\eqref{eq:last} follows from the maximum entropy property of the Gaussian distribution, which states that the differential entropy of a random vector with a given covariance is maximized when it is Gaussian with the same covariance matrix.
We assume that the symbol vector ${\bf s}$ is drawn from complex Gaussian distribution with mean zero and covariance $P{(N-L)K}{\bf I}_{(N-L)K}$. Hence, 
$$
\mathcal{H}({\bf s})=K(N-L)\log_2(\pi eP).
$$
The mutual information between ${\bf s}$ and $\hat{\bf s}$ satisfy:
$$
I({\bf s};\hat{\bf s})\geq  K(N-L)\log_2(\pi eP)-\log_2\left( \det(\pi e\, \mathbf{R}_{\tilde{\mathbf{s}}}) \right).
$$
The right-hand side of the inequality above defines a lower bound on the mutual information that holds for any type of estimator, regardless of its efficiency. Since the function $\log(\det(\cdot))$ is monotonically increasing over the cone of positive-definite matrices, this lower bound is maximized when the error covariance matrix $\mathbf{R}_{\tilde{\mathbf{s}}}$ is replaced by $\mathbb{E}_{\mathbf{s}}[\mathbf{CRB}(\mathbf{S}_d)]$. This observation motivates the definition of the following quantity:
\begin{align}
\mathcal{I}_{\rm joint} :=& K(N-L)\log_2(\pi e P) \notag\\
&-\log_2\left( \det(\pi e\, \mathbb{E}_{\mathbf{s}}[\mathbf{CRB}(\mathbf{S}_d)]) \right),
\end{align}
which serves as a lower bound on the mutual information in the case of an efficient estimator achieving the CRB, if such an estimator exists.

In the sequel, we analyze the asymptotic behavior of the normalized mutual information lower bound $\mathcal{I}_{\rm joint}/(K(N-L))$ under the following assumptions on the statistical properties of the channel matrix $\mathbf{G}$ and the growth rates of the system dimensions $M$, $K$, $L$, and $N$. This asymptotic regime models dense massive MIMO systems where the system dimensions increase jointly and multiple users are simultaneously served. Although the analysis is conducted in the large-system limit, the numerical results in Section~V demonstrate that the resulting power-allocation rules remain accurate for moderate system dimensions. 
\begin{assumption}
The channel matrix $\mathbf{G} \in \mathbb{C}^{M \times K}$ is modeled as
\[
\mathbf{G} = \mathbf{H} \mathbf{B}^{1/2},
\]
where $\mathbf{B} = \mathrm{diag}(\beta_1, \ldots, \beta_K)$ is a diagonal matrix representing the path-loss attenuation coefficients of the $K$ users, and $\mathbf{H}=[\mathbf{h}_1,\ldots,\mathbf{h}_K]$ is an $M \times K$ random matrix with independent and identically distributed (i.i.d.) complex Gaussian entries of zero mean and variance $1/M$. \label{ass:statistics}
\end{assumption}
\begin{assumption}
We assume that the number of users \(K\), the number of antennas \(M\), the number of pilot symbols \(L\), and the transmission block length \(N\) grow large at the same rate. Specifically, we consider the asymptotic regime where
\[
K, M, L, N \to \infty \,\,\, \text{with} \  \frac{K}{M} \to c, \ \frac{M}{N} \to \alpha, \text{and} \ \frac{L}{N} \to \beta,
\]
where the constants satisfy \(0 < c < 1\), \(0 < \beta < 1\), and \(\alpha > 0\).
\label{ass:grwoth_rate}
\end{assumption}
\begin{theorem}\label{theorem_MI_joint}
Under Assumption \ref{ass:statistics} and \ref{ass:grwoth_rate}, the quantity $\mathcal{I}_{\rm joint}$ satisfies the following convergence:
$$
\frac{\mathcal{I}_{\rm joint}}{K(N-L)}-\overline{\mathcal{I}}_{\rm joint}\xrightarrow[]{a.s.}0,
$$
where
\begin{align}
\overline{\mathcal{I}}_{\rm joint}=&\log_2\left(\frac{P}{e\sigma_v^2}\right)-\log_2\left(1+\frac{c\alpha P}{a\beta}\right)+\frac{1}{K}\sum_{k=1}^{K}\log_2(\beta_k)\nonumber\\
&+\frac{c-1}{c}\log_2(1-c).
\end{align}
\begin{proof}
    See Appendix~\ref{app:theorem_MI_joint}.
\end{proof}
\end{theorem}
To pave the way for optimizing the power and time resources between training and data transmission, we define the following quantity:
\begin{align}
\mathcal{T}_{\rm joint} := (1 - \beta) \overline{\mathcal{I}}_{\rm joint}, \label{eq:Tjoint}
\end{align}
which serves as a proxy for the achievable throughput. The factor \((1 - \beta)\) accounts for the fraction of symbols allocated to data transmission, reflecting the fact that a proportion \(\beta = L/N\) of the transmission block is dedicated to pilot training.

Based on the expression of the metric $\mathcal{T}_{\rm joint}$, we aim to derive an optimal power allocation strategy under the total average power constraint:
\[
P_{\rm total} = (1 - \beta)P + a\beta.
\]
To this end, let us define $x = a\beta$, which represents the average power per symbol allocated to training. Solving for $P$ yields $P = \frac{P_{\rm total} - x}{1 - \beta}$, allowing us to rewrite the throughput $\mathcal{T}_{\rm joint}$ as a function of $x$ and $\beta$:

\begin{align}
&\mathcal{T}_{\rm joint}(x,\beta) = (1 - \beta)\log_2\left(\frac{P_{\rm total} - x}{e\sigma_v^2(1 - \beta)}\right) \nonumber \\
&\quad - (1 - \beta)\log_2\left(1 + \frac{c\alpha(P_{\rm total} - x)}{x(1 - \beta)}\right) \nonumber \\
&\quad + \frac{1-\beta}{K}\sum_{k=1}^K \log_2(\beta_k) + (1 - \beta)\frac{c - 1}{c}\log_2(1 - c).
\end{align}
The optimization is carried out over the interval $x \in [0, P_{\rm total}]$, reflecting the feasible range for training power, and $\beta \in [0, 1]$, which captures the fraction of time (or symbols) devoted to training. The corresponding optimization problem is given by:
\begin{equation}
\max_{\substack{0 \leq x \leq P_{\rm total} \\ 0 \leq \beta \leq 1}} \overline{\mathcal{T}}_{\rm joint}(x, \beta).\label{eq:optimization_problem}
\end{equation}
Considering the optimization over $x$ for a given fixed $\beta$, we can prove that function $x\mapsto \mathcal{T}(x,\beta)$ has a unique maximizer given by:
\begin{equation}
x^\star(\beta)=\frac{P_{\rm total}\sqrt{c\alpha}}{\sqrt{1-\beta}+\sqrt{c\alpha}}. \label{eq:x_beta}
\end{equation}
By replacing $x$ by $x^\star(\beta)$, we obtain:
\begin{align}
\overline{\mathcal{T}}_{\rm joint}(x^\star(\beta),\beta)&=(1-\beta)\log_2(\frac{P_{\rm total}}{e\sigma_v^2})+\frac{1-\beta}{K}\sum_{k=1}^{K}\log_2(\beta_k)\nonumber \\
&-2(1-\beta )\log_2(\sqrt{1-\beta}+\sqrt{c\alpha})\nonumber\\
&+(1-\beta)\frac{c-1}{c}\log_2(1-c).
\end{align}
By taking the second derivative of $g:\beta\mapsto\mathcal{T}_{\rm joint}(x^\star(\beta),\beta) $ with respect to $\beta$, we can easily check that it is strictly concave. Hence it admits a unique maximizer $\beta^\star$. Let us compute the derivative of $g$ with respect to $\beta$. We obtain:
\begin{align}
g'(\beta)=&-\log_2(\frac{P_{\rm total}}{e\sigma_v^2})+2\log_2(\sqrt{1-\beta}+\sqrt{c\alpha})\nonumber \\
&-\frac{1}{K}\sum_{k=1}^{K}\log_2(\beta_k)-\frac{c-1}{c}\log_2(1-c)\nonumber \\
&+\frac{\sqrt{1-\beta}}{(\ln 2)(\sqrt{1-\beta}+\sqrt{c\alpha})}.
\end{align}
By replacing $\beta$ by $0$, we note that the derivative of $g(\beta)$ at $\beta=0$ simplifies to:
\begin{align}
g'(0)=&\frac{\ln(\frac{e\sigma_v^2(1+\sqrt{c\alpha})^2}{P_{\rm total}})}{\ln 2}-\frac{1}{K}\sum_{k=1}^{K}\log_2(\beta_k)\nonumber \\
&+\frac{1}{(\ln 2)(1+\sqrt{c\alpha})}-\frac{c-1}{c}\log_2(1-c).
\end{align}
Let $\theta$ be given by:
$$
\theta=\frac{1}{K}\sum_{k=1}^{K}\ln(\beta_k)-\frac{1}{(1+\sqrt{c\alpha})}+\frac{c-1}{c}\ln(1-c).
$$
Then, if $\sigma_v^2\leq \sigma_{\rm th}^2:=\frac{\exp(\theta) P_{\rm total}}{e(1+\sqrt{c\alpha})^2}$, $g'(0)\leq 0$ and hence, due to the concavity of $g$, $\beta=0$ is the optimum of $\overline{\mathcal{T}}_{\rm joint}(x^\star(\beta),\beta)$. Since $\beta = 0$ is the optimum of $\overline{\mathcal{T}}(x^\star(\beta), \beta)$, substituting $\beta = 0$ into \eqref{eq:x_beta} yields  
\begin{equation}
x^\star(0) = \frac{P_{\rm total}\sqrt{c\alpha}}{1+\sqrt{c\alpha}}.
\end{equation}
This power allocation scheme should not be interpreted as using zero training symbols, since the power dedicated to training remains non-zero. Rather, it corresponds to a limiting case: when the noise power is sufficiently low, the optimal strategy according to the defined metric is to use very few training symbols with relatively high power, while allocating the majority of the transmission resources to data transmission.

We assume now that the noise power satisfies: $\sigma_v^2> \frac{2^\theta P_{\rm total}}{e(1+\sqrt{c\alpha})}$. Since $\lim_{\beta\to 1}g'(\beta)\to -\infty$, there is a unique $\beta^\star$ satisfying 
$$
g'(\beta^\star)=0.
$$
Writing the first order optimality condition, we obtain that $\beta^\star$ is the unique solution to:
\begin{equation}
\ln(\sqrt{1-\beta}+\sqrt{c\alpha})-\frac{\rho}{2}-\frac{1}{2}\frac{\sqrt{c\alpha}}{\sqrt{1-\beta}+\sqrt{c\alpha}}=0. \label{eq:to_be_solved}
\end{equation}
where 
$$
\rho=\frac{1}{K}\sum_{k=1}^{K}\ln \beta_k+\frac{c-1}{c}\ln(1-c)+\ln \frac{P_{\rm total}}{e\sigma_v^2}-1.
$$
Let $y=\frac{1}{\sqrt{1-\beta}+\sqrt{c\alpha}}$. Then, with this change of variable \eqref{eq:to_be_solved} becomes:
$$
\ln y +\frac{\sqrt{c\alpha}}{2}y=-\frac{\rho}{2}.
$$
or equivalently
$$
\frac{\sqrt{c\alpha}}{2}y\exp\left(\frac{y\sqrt{c\alpha}}{2}\right)=\frac{\sqrt{c\alpha}}{2}\exp\left(\frac{-\rho}{2}\right).
$$
Hence, 
$$
\frac{\sqrt{c\alpha}}{2}y=W\left(\frac{\sqrt{c\alpha}}{2}\exp(\frac{-\rho}{2})\right).
$$
where $W$ is Lambert W function~\cite{bjornson2015optimal}. We can now reformulate the above equation using the original variable $\beta$ to obtain:
\begin{equation}
\beta^\star=1-\left(\frac{1}{\frac{2}{\sqrt{c\alpha}}W(\exp(-\frac{\rho}{2})\frac{\sqrt{c\alpha}}{2})}-\sqrt{c\alpha}\right)^2. \label{eq:beta_star_eq}
\end{equation}
The following theorem summarizes the optimal power allocation that maximizes \eqref{eq:optimization_problem}.
\begin{theorem}\label{joint_power_allocation}
The optimal power allocation $x^\star$ and $\beta^\star$ that maximizes \eqref{eq:optimization_problem} is given by:
\begin{enumerate}
\item If $\sigma_v^2\leq \sigma_{\rm th}^2$, then $\beta^\star=0$ and $x^\star=\frac{P_{\rm total}\sqrt{c\alpha}}{1+\sqrt{c\alpha}}$.
\item If $\sigma_v^2> \sigma_{\rm th}^2$, then $\beta^\star$ is given by \eqref{eq:beta_star_eq}
and $x^\star=\frac{P_{\rm total}\sqrt{c\alpha}}{\sqrt{1-\beta^\star}+\sqrt{c\alpha}}$.
\end{enumerate}
\end{theorem}
Let us now consider the regime in which the number of users is fixed while the other dimensions grow large with the same pace. This regime reflects practical scenarios where only a limited number of users are scheduled within each coherence block due to pilot overhead, scheduling constraints, or latency requirements, even when the BS is equipped with a large antenna array. Particularly, we consider the following growth rate regime. 
\begin{assumption}
The number of users $K$ is fixed, while all other dimensions $M,N$ and $L$ grow large with the same pace as specified in Assumption \ref{ass:grwoth_rate}. 
\label{ass:growth_rate_1}
\end{assumption}
\begin{theorem}
Under Assumption \ref{ass:statistics} and \ref{ass:growth_rate_1}, the quantity $\mathcal{I}_{\rm joint}$ satisfies the following convergence:
$$
\frac{\mathcal{I}_{\rm joint}}{K(N-L)}- \overline{\mathcal{I}}_{\rm joint}^{\rm fixed}\xrightarrow[]{a.s.} 0,
$$
with 
\begin{align}
\overline{\mathcal{I}}_{\rm joint}^{\rm fixed}:=\frac{1}{K}\sum_{k=1}^{K}\log_2(\frac{\beta_k P}{e\sigma_v^2}).
\end{align}
\end{theorem}
Based on the expression of $\overline{\mathcal{I}}_{\rm joint}^{\rm fixed}$, we define the following quantity:
\begin{align}
\overline{\mathcal{J}}_{\rm joint}^{\rm fixed}&=(1-\beta)\overline{\mathcal{I}}_{\rm joint}^{\rm fixed} \notag\\
&=(1-\beta)\sum_{k=1}^{K}\log_2(\beta_k\frac{P_{\rm total}-x}{(1-\beta)e\sigma_v^2}),
\end{align}
where the second equality follows by replacing $P$ by $\frac{P_{\rm total}-x}{1-\beta}$. We note that $\overline{\mathcal{J}}_{\rm joint}^{\rm fixed}$ is a decreasing function with respect to $x$. However, this does not mean that the optimal choice in practice is to allocate zero power to training symbols. Instead, it merely indicates that in the asymptotic regime, allocating a large fraction of the total power to training is not beneficial. Assuming $x$ is fixed, the above expression allows us to optimize with respect to  $\beta$. After some simple calculations, we obtain:
$$
\beta^\star=\max\Big(0,1-\frac{P_{\rm total}-x}{e\sigma_v^2}\big(\prod_{k=1}^{K}\beta_k\big)^{\frac{1}{K}}\Big).
$$
This expression is only operative when
$\frac{P_{\rm total}-x}{e\sigma_v^2}\big(\prod_{k=1}^{K}\beta_k\big)^{\frac{1}{K}}<1$
which guarantees the existence of an interior optimum for $\beta$.

\section{Sequential channel and symbol estimation}
\label{sec:sequential}
In this section, we consider Strategy (ii), namely sequential channel and data symbol estimation. The procedure consists of two stages. In the first stage, the channel is estimated using the LMMSE estimator based solely on the training transmissions. In the second stage, this channel estimate is used to construct an effective system model for data detection, from which the BCRB for the symbol matrix is derived. The adoption of the BCRB is motivated by the statistical nature of the channel uncertainty in the sequential strategy. Specifically, after the training phase, the channel is naturally modeled as a random parameter characterized by its posterior distribution. Therefore, a Bayesian performance bound that incorporates the prior channel statistics is more appropriate than the conventional deterministic CRB. This formulation is reminiscent of conventional training-based schemes, where channel estimation relies exclusively on the signals transmitted during the training phase and subsequent data detection is performed conditioned on the estimated channel. 

\subsection{Channel estimation using the LMMSE}
Let ${\bf z}_k = [s_k(0), \ldots, s_k(L-1)]$ denote the training sequence assigned to user $k$.  
By correlating the received training matrix with $\frac{1}{aL} {\bf z}_k$, we obtain  
\begin{equation}
\frac{1}{aL} {\bf Y}_p {\bf z}_k^{H} = {\bf g}_k + \frac{1}{aL} {\bf V}_p {\bf z}_k^{H}.
\end{equation}
We assume that the channel matrix ${\bf G}$ satisfies Assumption~\ref{ass:statistics}, namely, its columns ${\bf g}_1, \ldots, {\bf g}_K$ are independent and obey $\mathbb{E}_{\mathbf{h}_k}[{\bf g}_k {\bf g}_k^{H}] = \frac{\beta_k}{M} {\bf I}_M$. Under this assumption, the LMMSE estimate of ${\bf g}_k$ is given by~\cite{kay1993fundamentals,kailath2000linear}  
\begin{equation}
\hat{\bf g}_k = \frac{\beta_k}{\beta_k + \frac{M\sigma_v^2}{aL}} \cdot \frac{1}{aL} {\bf Y}_p {\bf z}_k^{H},
\end{equation}
which follows a zero-mean complex Gaussian distribution with covariance  
\begin{equation}\label{covariance_hat_g}
\mathbb{E}_{\mathbf{h}_k,\mathbf{V}_p}[\hat{\bf g}_k \hat{\bf g}_k^{H}] = 
\frac{\beta_k^2}{M\left(\beta_k + \frac{M\sigma_v^2}{aL}\right)} \, {\bf I}_M.
\end{equation}
By the orthogonality property of the LMMSE estimator, the channel vector can be decomposed as 
\begin{equation}
{\bf g}_k = \hat{\bf g}_k + \tilde{\bf g}_k,
\end{equation}
where the estimation error vectors $\tilde{\bf g}_k$ are independent  and are  zero-mean complex Gaussian vector with covariance
\begin{equation}
\mathbb{E}_{\mathbf{h}_k,\mathbf{V}_p}[\tilde{\bf g}_k \tilde{\bf g}_k^{H}] = 
\frac{\beta_k \frac{\sigma_v^2}{aL}}{\beta_k + \frac{M\sigma_v^2}{aL}} \, {\bf I}_M.
\end{equation}
Let $\hat{\bf G}=[\hat{\bf g}_1,\cdots,\hat{\bf g}_K]$ be the estimated channel matrix. The signal received during data transmission can be expressed as
\begin{equation}
{\bf y}(n)=\hat{\bf G}{\bf s}(n)+\boldsymbol{\omega}(n), \label{eq:received_sequential}
\end{equation}
where $\boldsymbol{\omega}(n)=\tilde{\bf G}{\bf s}(n)+{\bf v}(n)$. Conditioning on ${\bf s}(n), n=L,\cdots, N-1$, $\{\boldsymbol{\omega}(n)\}_{n=L}^{N-1}$ is a Gaussian process with mean zero and auto-covariance:
\begin{align}
{\bf R}_\omega(n,m)&= \mathbb{E}[\boldsymbol{\omega}(n)\boldsymbol{\omega}(m)^{H}]\nonumber \notag\\
&=\mathbb{E}[\tilde{\bf G}{\bf s}(n){\bf s}(m)^{H}\tilde{\bf G}^{H}] +\mathbb{E}[{\bf v}(n){\bf s}(m)^{H}\tilde{\bf G}^{H}]\nonumber \notag\\
&\quad+\mathbb{E}[\tilde{\bf G}{\bf s}(n){\bf v}(m)^{H}]+\mathbb{E}[{\bf v}(n){\bf v}(m)^{H}] \notag\\
&= \mathbb{E}\left[\sum_{i=1}^{K}\sum_{j=1}^{K}{\tilde{\bf g}}_is_i(n)s_j^\ast(m)\tilde{\bf g}_j^{H}\right]+\sigma_v^2\delta_{n=m}{\bf I}_{M} \notag\\
&=\sum_{i=1}^K s_i(n)s_i^\ast(m)\frac{\beta_i\sigma_v^2 {\bf I}_M}{aL(\beta_i+\frac{M}{aL}\sigma_v^2)}+\sigma_v^2\delta_{n=m}{\bf I}_{M}. \notag
\end{align}
\begin{remark}
The effective noise term $\boldsymbol{\omega}(n)$ is symbol-dependent and temporally correlated due to the common channel estimation error component $\tilde{\bf G}$. As shown above, its second-order statistics explicitly depend on the transmitted symbols through the terms $s_i(n)s_i^\ast(m)$. In the subsequent analysis, we characterize this effective noise through its covariance structure and adopt an LMMSE-based sequential detection framework. While more sophisticated detectors could, in principle, exploit the full temporal and symbol-dependent structure of $\boldsymbol{\omega}(n)$, incorporating such receiver-specific processing would significantly increase analytical complexity and obscure the main objective of this work. Our goal is not to establish detector optimality under colored noise, but rather to quantify how channel estimation uncertainty propagates into symbol-level performance and affects training–data power allocation. The proposed CRB-based and mutual-information-related performance proxies capture this second-order impact and provide tractable design guidance, as validated by the finite-dimensional numerical results in Section~\ref{sec_sim}.    
\end{remark}
\subsection{Symbol CRB matrix} We consider the estimation problem of the symbol matrix ${\bf S}_d$ based on the received symbols ${\bf y}(n), n=L,\cdots, N-1$ modeled as in 
\eqref{eq:received_sequential}. Contrary to the previous case, we assume here that the parameter we would like to estimate ${\bf s}={\rm vec}({\bf S}_d)$ is standard complex gaussian vector with mean zero and covariance $P{\bf I}_{(N-L)K}$, while the channel $\hat{\bf G}:=[\hat{\bf g}_1,\cdots,\hat{\bf g}_K]$ is assumed to be deterministic. Let $\hat{\bf s}$ denote an estimate of ${\bf s}$, then the BCRB satisfies:
$$
\mathbb{E}[({\bf s}-\hat{\bf s})({\bf s}-\hat{\bf s})^{H}]\succeq {\bf BCRB},
$$
where 
\begin{align}
&{\bf BCRB}^{-1}\notag\\
&=\left(\mathbb{E}_{{\bf y},{\bf s}|\hat{\bf G}}\left[\left(\frac{\partial \ln p({\bf y},{\bf s}|\hat{\bf G})}{\partial {\bf s}^\ast}\right)\left(\frac{\partial \ln p({\bf y},{\bf s}|\hat{\bf G})}{\partial {\bf s}^\ast}\right)^{H}\right]\right)^{-1}. \notag
\end{align}

\begin{theorem}[Derivation of the BCRB]For sequential channel and symbol estimation, the BCRB for symbol estimates is given by:
\begin{align}
{\bf BCRB}=&\Big(\mathbb{E}_{{\bf s}}\left[{\bf A}^{-1}\otimes \left(\hat{\bf G}^{H}\hat{\bf G}+\frac{M}{L^2}{\bf D}{\bf S}_d{\bf A}^{-T}{\bf S}_d^{H}{\bf D}\right)\right] \notag\\
 &+P^{-1}{\bf I}_{K(N-L)}\Big)^{-1},
\end{align}
where ${\bf A}=\frac{1}{L}{\bf S}_d^{T}{\bf D}{\bf S}_d^\ast+\sigma_v^2{\bf I}_{N-L}$ represents the effective covariance matrix of the data symbols in the sequential estimation stage, accounting for both the data symbols and the residual channel estimation uncertainty. 
The diagonal matrix ${\bf D}$ captures the impact of the training-based channel estimation error, where each diagonal entry reflects the reliability of the estimated channel for the corresponding user and is given by
\begin{align}
{\bf D}
=\mathrm{diag}\left\{\left[\frac{\beta_1\sigma_v^2}{a\beta_1+\frac{M}{L}\sigma_v^2},\ldots,\frac{\beta_K\sigma_v^2}{a\beta_K+\frac{M}{L}\sigma_v^2}\right]\right\}.
\end{align}
\label{sequential_bcrb}
\end{theorem}
\begin{proof}
See Appendix~\ref{app:sequential_bcrb}.
\end{proof}

\par

\subsection{Mutual Information for Gaussian Symbols}
Similar to the joint channel and symbol estimation scenario, we may use the same calculations as in \eqref{eq:last} to motivate introducing the following quantity
$$
\mathcal{I}_{\rm seq}:=K(N-L)\log_2(\pi eP)-\log_2({\rm det}(\pi e {\bf BCRB})),
$$
This metric constitutes a lower bound on the mutual information when an efficient estimator achieving the BCRB exists. In the following, we analyze the asymptotic behavior of $\mathcal{I}_{\rm seq}/(K(N-L))$ under the regimes specified in Assumptions~\ref{ass:statistics} and~\ref{ass:grwoth_rate}. The derivation relies on a large-system approximation of the BCRB expression together with standard RMT results for Gaussian signaling. This leads to the following theorem.
\begin{theorem}[Asymptotic derivation of $\overline{\mathcal{I}}_{\rm seq}$] 
Assume the data symbol matrix ${\bf S}_d$ is composed of independent and identically distributed complex Gaussian symbols with mean zero and variance $P$. Then, under Assumption \ref{ass:statistics} and \ref{ass:grwoth_rate}, the quantity $\mathcal{I}_{\rm seq}$ converges to: \label{theorem_MI_sequential}
\begin{align}
    \frac{\mathcal{I}_{\rm seq}}{K(N-L)} - \bar{\mathcal{I}}_{\rm seq} \xrightarrow[]{a.s.}0,
\end{align}
where 
\begin{align}
    \bar{\mathcal{I}}_{\rm seq}:=&\log_2\Big(\frac{P}{\sigma_v^2(1+\overline{m})}\Big)+\frac{1}{K}\log_2\Big({\rm det}\Big(\frac{{\bf R}_{\hat{\bf G}}}{1+\kappa}\nonumber\\
    &+\frac{\alpha(1-\beta)P}{(1+\overline{m})}\,\tilde{\bf D}^2{\bf T}+\frac{\sigma_v^2(1+\overline{m})}{P}{\bf I}_K\Big)\Big) \notag\\
    &+\frac{1}{c}\log_2(1+\kappa)-\frac{\kappa}{c(1+\kappa)\ln(2)}, 
\end{align}
with $x=a\beta$ and
\begin{align}
    &{\bf R}_{\hat{\bf G}} = {\bf B}^2\Big({\bf B}+\sigma_v^2\frac{\alpha}{x}{\bf I}_K\Big)^{-1}, \\
    &{\bf T}=\Big(\sigma_v^{2}{\bf I}_K+\frac{(1-\beta)P}{(1+\overline{m})}\tilde{\bf D}\Big)^{-1},
\end{align}
where $\tilde{\bf D}$ denotes the deterministic equavalent of ${\bf D}/\beta$ given by 
\begin{equation}
\tilde{\bf D}:={\rm diag}\left\{\left[\frac{\beta_1\sigma_v^2}{x\beta_1+\alpha\sigma_v^2},\cdots,\frac{\beta_K\sigma_v^2}{x\beta_K+\alpha\sigma_v^2}\right]\right\},\label{eq:tilde_D}\end{equation}
and $\overline{m}$ is the unique positive solution in $m$ to the following equation:
$$
\overline{m}=\frac{1}{N}{\rm tr}\left(\tilde{\bf D}\Big({\frac{\sigma_v^2}{P}{\bf I}_K+\frac{1-\beta}{(1+\overline{m})}}\tilde{\bf D}\Big)^{-1}\right),
$$
and $\kappa$ is the unique solution to the following equation:
\begin{align}
\kappa =&\frac{1}{M}{\rm tr}\Big({\bf R}_{\hat{G}}\Big(\frac{{\bf R}_{\hat{G}}}{1+\kappa} +\frac{\alpha(1-\beta)P}{(1+\overline{m})}\,\tilde{\bf D}^2{\bf T} \notag\\
&+\frac{\sigma_v^2(1+\overline{m})}{P}{\bf I}_K\Big)^{-1}\Big). 
\end{align}

\begin{proof}
    See Appendix \ref{app:theorem_MI_sequential}. 
\end{proof}
\end{theorem}

\par

\noindent{\bf On the case of equal fading coefficients}
Let us consider the case where all large scale fading coefficient are the same, i.e., $\beta_1=\beta_2=\cdots=\beta_K=\delta$. In this case $\tilde{\bf D}=\frac{\sigma_v^2\delta}{x\delta+\alpha\sigma_v^2}{\bf I}_K:=d{\bf I}_K$. Furthermore, 
${\bf T}$ and $\hat{\bf R}_{\hat{\bf G}}$ become:
\begin{align}
&{\bf T}=\frac{(1+\overline{m})}{(1+\overline{m})\sigma_v^2+(1-\beta)P d}{\bf I}_K,  \\
&{\bf R}_{{\hat{\bf G}}}=\frac{\delta^2}{\delta+\frac{\alpha}{x}\sigma_v^2}{\bf I}_K=\frac{x\delta{d}}{ \sigma_v^2}{\bf I}_K,
\end{align}
where $\overline{m}$ is the unique positive solution to the following second-order polynomial:
$$
\overline{m}^2+(1+(1-\beta-c\alpha)P\frac{d}{\sigma_v^2}\overline{m}-\frac{c\alpha P}{\sigma_v^2}d=0. 
$$
Similarly, $\kappa$ is the unique positive solution to the following second order polynomial:
$$
A\kappa^2+\Big(A+\frac{(1-c)x\delta d}{\sigma_v^2}\Big)\kappa -\frac{cx\delta }{\sigma_v^2}d=0,
$$
where      
\begin{align}
A =& \frac{\alpha(1-\beta)Pd^2}{(1+\overline{m})\sigma_v^2+(1-\beta)Pd} + \frac{\sigma_v^2(1+\overline{m})}{P} \notag\\
     =& \frac{(1-\beta)d\overline{m}}{c(1+\overline{m})} + \frac{\sigma_v^2(1+\overline{m})}{P}. \notag 
\end{align}
Based on this, we can simplify $\overline{\mathcal{I}}_{\rm seq}$ as:
\begin{align}
\overline{\mathcal{I}}_{\rm seq}:=&\log_2(\frac{P}{\sigma_v^2(1+\overline{m})})+\log_2\Big(\frac{x\delta d}{(1+\kappa)\sigma_v^2}\nonumber \\
&+\frac{\alpha(1-\beta)Pd^2}{(1+\overline{m})\sigma_v^2+(1-\beta )Pd )}+\frac{\sigma_v^2(1+\overline{m})}{P}\Big) \nonumber\\
&+\frac{1}{c}\log_2(1+\kappa)-\frac{\kappa}{c(1+\kappa)\ln 2} \nonumber\\
=&\log_2(\frac{P}{\sigma_v^2(1+\overline{m})})+\log_2\big(\frac{x\delta^2c}{(x\delta+\alpha\sigma_v^2)\kappa}\big) \nonumber\\
&+\frac{1}{c}\log_2(1+\kappa)-\frac{\kappa}{c(1+\kappa)\ln(2)}. 
\end{align}
Using the asymptotic expression of $\overline{\mathcal{I}}_{\rm seq}$, we introduce the quantity
\begin{equation}
\overline{\mathcal{T}}_{\rm seq} := (1-\beta)\,\overline{\mathcal{I}}_{\rm seq},\label{eq:Tseq}
\end{equation}
which acts as a proxy for the achievable throughput under the sequential symbol estimation strategy.
To design an optimal power allocation policy under the total average power constraint
\begin{equation}
    P_{\rm total} = (1-\beta)P + a\beta = (1-\beta)P + x,
\end{equation}
we consider the optimization problem
\begin{equation}
    \max_{\substack{0\leq x \leq P_{\rm total}\\[1mm] 0\leq \beta \leq 1}}
    \ \ \overline{\mathcal{T}}_{\rm seq}.
\end{equation}

In contrast with the joint channel--symbol estimation case, obtaining closed-form expressions 
for the optimal $x$ and $\beta$ appears difficult. Nevertheless, 
$\overline{\mathcal{T}}_{\rm seq}$ can be efficiently evaluated for any $(\beta,x)$, 
so the optimal pair can be determined numerically without difficulty.

To continue, we consider next the convergence of $\mathcal{I}_{\rm seq}/(K(N-L))$ under the regime specified in Assumption~\ref{ass:growth_rate_1}. 
\begin{theorem}
 Under assumption \ref{ass:growth_rate_1}, the following convergence hold true:
$$
\frac{\mathcal{I}_{\rm seq}}{K(N-L)}-\overline{\mathcal{I}}_{\rm seq}^{\rm fixed}\xrightarrow[a.s.]{} 0,
$$
with
\begin{align}
&\overline{\mathcal{I}}_{\rm seq}^{{\rm fixed}}=\frac{1}{K}\log_2({\rm det}(\frac{P}{\sigma_v^2}{\bf R}_{\hat{\bf G}}\nonumber \\
&\qquad+\frac{\alpha(1-\beta)P}{\sigma_v^2}\tilde{\bf D}^2(\frac{\sigma_v^2}{P}{\bf I}_K+(1-\beta) \tilde{\bf D})^{-1}+{\bf I}_K)).
\end{align}
\end{theorem}
\begin{proof}
The proof relies on the fact that in the regime of Assumption \ref{ass:growth_rate_1}, 
$$
\frac{1}{N-L}{\rm tr}({\bf A}^{-1})\xrightarrow[]{a.s.}\frac{1}{\sigma_v^2},
$$
and 
\begin{align}
&\hat{\bf G}^{H}\hat{\bf G}+\frac{M}{L^2}{\bf D}{\bf S}_d{\bf A}^{-T}{\bf S}_d^{H}{\bf D}\nonumber \\
&-\left({\bf R}_{\hat{\bf G}}
+\alpha(1-\beta)P\,\tilde{\bf D}^2
\big(\sigma_v^2{\bf I}_K+(1-\beta)P \tilde{\bf D}\big)^{-1}\right)
\xrightarrow[]{a.s.}0. \notag
\end{align}
\end{proof}
By replacing $P$ by $\frac{P_{\rm total}-x}{1-\beta}$, we can expand $\overline{\mathcal{I}}_{\rm seq}^{\rm fixed}$ as:
\begin{align}
\overline{\mathcal{I}}_{\rm seq}^{\rm fixed}=&\frac{1}{K}\sum_{k=1}^{K}\log_2\Big(\frac{P}{\sigma_v^2}\frac{\beta_k^2}{\beta_k+\frac{\sigma_v^2\alpha}{x}}\nonumber \\
&+\frac{\alpha(1-\beta)P\beta_k^2}{(x\beta_k+\alpha\sigma_v^2)(\frac{1}{P}(x\beta_k+\alpha\sigma_v^2)+(1-\beta)\beta_k)}+1\Big)  \notag\\
=&\frac{1}{K}\sum_{k=1}^{K}\log_2(\frac{(P_{\rm total}-x)P_{\rm total}\beta_k^2}{(1-\beta)(\alpha\sigma_v^2+P_{\rm total}\beta_k)}+1).
\end{align}
Based on the expression 
$
\overline{\mathcal{I}}_{\rm seq}^{\rm fixed},
$ we define the following quantity:
$$
\overline{\mathcal{T}}_{\rm seq}^{\rm fixed}:=(1-\beta)\overline{\mathcal{I}}_{\rm seq}^{\rm fixed}.
$$
Similar to the joint channel--symbol estimation strategy, the asymptotic limit of 
$\mathcal{I}_{\rm seq}/(K(N-L))$ is a non-increasing function of $x$ under the growth 
regime specified in Assumption~\ref{ass:growth_rate_1}. 
This does not imply, in practice, that we should allocate zero power to training symbols. 
Rather, it indicates that in the asymptotic regime, a larger fraction of the available 
power should be devoted to data symbols. We consider thus optimizing $\overline{\mathcal{T}}_{\rm seq}^{\rm fixed}$ with respect to $\beta$ for fixed $x$. This can be done numerically. However, a closed-form expression can be obtained when all the large-scale fading coefficients $\beta_k,k=1,\cdots,K$ are equal to the same quantity. 

\noindent{\bf On the case of equal fading coefficients.} When all $\beta_1=\cdots=\beta_K=\delta$. In this case, $\overline{\mathcal{T}}_{\rm seq}^{\rm fixed}$ simplifies to:
\begin{align}
\overline{\mathcal{T}}_{\rm seq}^{\rm fixed}&=(1-\beta)\log_2\Big(\frac{\delta^2P_{\rm total}(P_{\rm total}-x)}{\sigma_v^2(\alpha \sigma_v^2+\delta P_{\rm total})(1-\beta)}+1\Big) \notag\\
&=(1-\beta)\log_2\Big(\frac{\Theta}{(1-\beta)}+1\Big),
\end{align}
where
$$
\Theta=\frac{\delta^2P_{\rm total}(P_{\rm total}-x)}{\sigma_v^2(\alpha \sigma_v^2+\delta P_{\rm total})}.
$$
It can be readily verified that $\overline{\mathcal{T}}_{\rm seq}^{\rm fixed}$ is a non-increasing function of $\beta$ for $\beta \in (0,1)$, which directly implies that the optimal value is  $\beta^\star = 0$. This conclusion should not be interpreted as prescribing zero training symbols in practice; rather, it indicates that within this asymptotic regime, both the number of training symbols and their associated power must remain small relative to the total transmission resources.

\par 

\section{Simulation results}\label{sec_sim}
This section provides numerical simulations illustrating the asymptotic behaviors of strategies (i) and (ii) described in Section \ref{sec_JCSE} and Section \ref{sec:sequential}. In the simulation experiments, users are uniformly distributed within the single cell with radius $500$\,m, respectively. For large-scale fading coefficients, we consider the normalized three-slope path loss model~\cite{tang2001mobile} as follows
\begin{align}
    \beta_k = \left\{\begin{array}{ll}
       c_0  & d_k \le d_0 \\
       \frac{c1}{d_k^2}  & d_0<d_k\le d_1 \\
       \frac{c_2z_k}{d_k^{3.5}} & d_k >d_1
    \end{array}\right. ,
\end{align}
where $d_k$ is the distance between the BS and $k$th user, $z_k$ is the log-normal shadow fading with zero mean and variance $8$\,dB, and 
\begin{align}
    10\log_{10}c_2 =& -46.3-33.9\log_{10}f + 13.82\log_{10}h_B \notag \\
    &+ (1.1\log_{10}f-0.7)h_R - (1.56\log_{10}f-0.8), \notag\\
    10\log_{10}c_1 =& 10\log_{10}c_2 - 15\log_{10}(d_1), \notag \\
    10\log_{10}c_0 =& 10\log_{10}c_1 - 15\log_{10}(d_0),\notag
\end{align}
with $d_0 = 10$\,m, $d_1=50$\,m, $f=1900$\,MHz being the carrier frequency, $h_B=15$\,m being the BS antenna height, $h_R$ being the user antenna height. Additionally, the signal-to-noise ratio (SNR) is fixed at $10$\,dB and calculated by $\mathrm{SNR}=\frac{P_{\mathrm{total}}}{\sigma_v^2}$. We assume that the training   ${\bf S}_p$ has orthogonal rows satisfying ${\bf S}_p{\bf S}_p^{H}=aL{\bf I}_K$ and ${\bf S}_d$ consists of independent and identically distributed complex Gaussian entries with zero mean and variance $P$.
We conduct two types of experiments. In the first set, we assess the accuracy of the asymptotic approximations to the lower bound on the mutual information under the regime specified in Assumption~\ref{ass:grwoth_rate}. To this end, we introduce the normalized lower bound approximation error (NLAE),
\begin{align}
\mathrm{NLAE}
= \frac{\mathbb{E}\left[|\frac{\mathcal{I}}{K(N-L)} - \overline{\mathcal{I}}|^2\right]}
        {\mathbb{E}\left[|\frac{\mathcal{I}}{K(N-L)}|^2\right]},
\end{align}
where $\mathcal{I}$ denotes the true lower bound for strategy~(i), $\mathcal{I}_{\mathrm{joint}}$, or strategy~(ii), $\mathcal{I}_{\mathrm{seq}}$, and $\overline{\mathcal{I}}$ denotes its corresponding asymptotic expression, $\overline{\mathcal{I}}_{\mathrm{joint}}$ or $\overline{\mathcal{I}}_{\mathrm{seq}}$.
In the second set of experiments, we evaluate the practical relevance of the power allocation guidelines derived for both strategies within the same asymptotic regime.

\subsection{Accuracy assessment of the asymptotic expressions $\overline{\mathcal{I}}_{\rm joint}$ and $\overline{\mathcal{I}}_{\rm seq}$}
In the first experiment, we assess the accuracy of the derived lower bounds on the mutual information for Gaussian symbols. We fix $P_{\mathrm{total}} = 1$. The block length is chose as $N\in\{16,32,64,128,256,512,1024,2048\}$, which covers small to moderate system dimensions of practical interest. Figures~\ref{MI_SNR_joint} and~\ref{MI_seq_SNR} display the empirical values of $\mathcal{I}_{\rm joint}/(K(N-L))$ and $\mathcal{I}_{\rm seq}/(K(N-L))$, respectively, as functions of $N$, with the parameters $\alpha$, $\beta$, and $c$ all set to $1/2$, and for different SNR values. The curves clearly show that the proposed asymptotic expressions provide an excellent match to the simulated lower bounds for both the joint and sequential channel--and--symbol estimation strategies, thereby validating the accuracy of the asymptotic characterization.

Next, we compare the NLAE between the proposed asymptotic behavior and its counterpart obtained under Assumption~\ref{ass:grwoth_rate} for strategies (i) and (ii). Fixing the SNR at $10$\,dB and setting $c = 1/2$, the results in Figures~\ref{NLAE_joint_parameters} and~\ref{NLAE_seq_parameters} show that, for both estimation strategies, the approximation error of the proposed expressions steadily decreases as the system dimensions increase. This behavior is consistent across all considered values of $\alpha$ and $\beta$, further confirming the asymptotic accuracy and robustness of the derived analytical results.

\begin{figure}[t]
\centering
\includegraphics[width=1.0\linewidth]{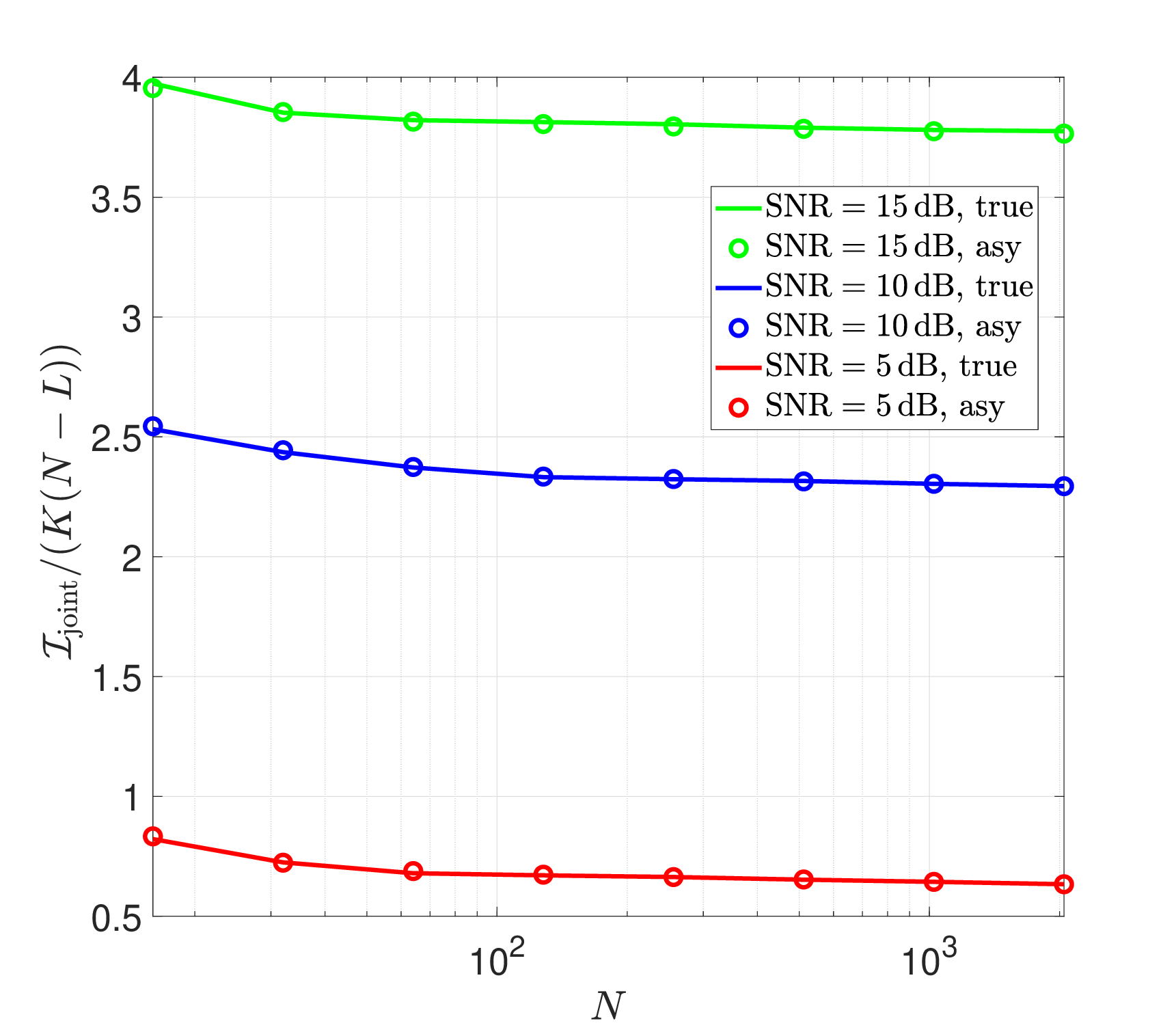}
\caption{$\mathcal{I}_{\mathrm{joint}}/(K(N-L))$ versus $N$ under different SNR levels for strategy (i) with $\alpha=\beta=c=1/2$.}
\label{MI_SNR_joint}
\end{figure}

\begin{figure}[t]
\centering
\includegraphics[width=1.0\linewidth]{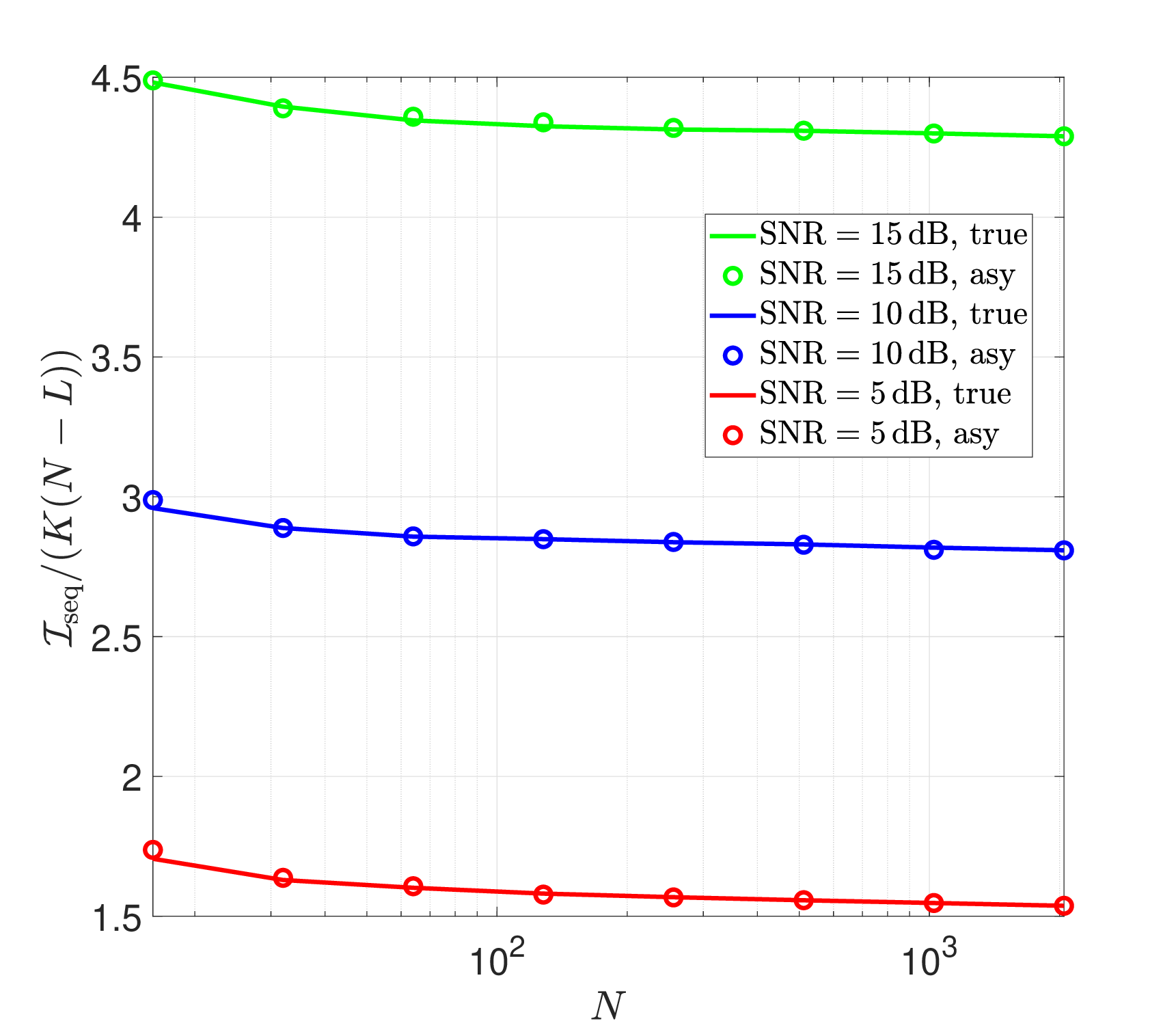}
\caption{$\mathcal{I}_{\mathrm{seq}}/(K(N-L))$ versus $N$ under different SNR levels for strategy (ii) with $\alpha=\beta=c=1/2$.}
\label{MI_seq_SNR}
\end{figure}

\begin{figure}[t]
\centering
\includegraphics[width=1.0\linewidth]{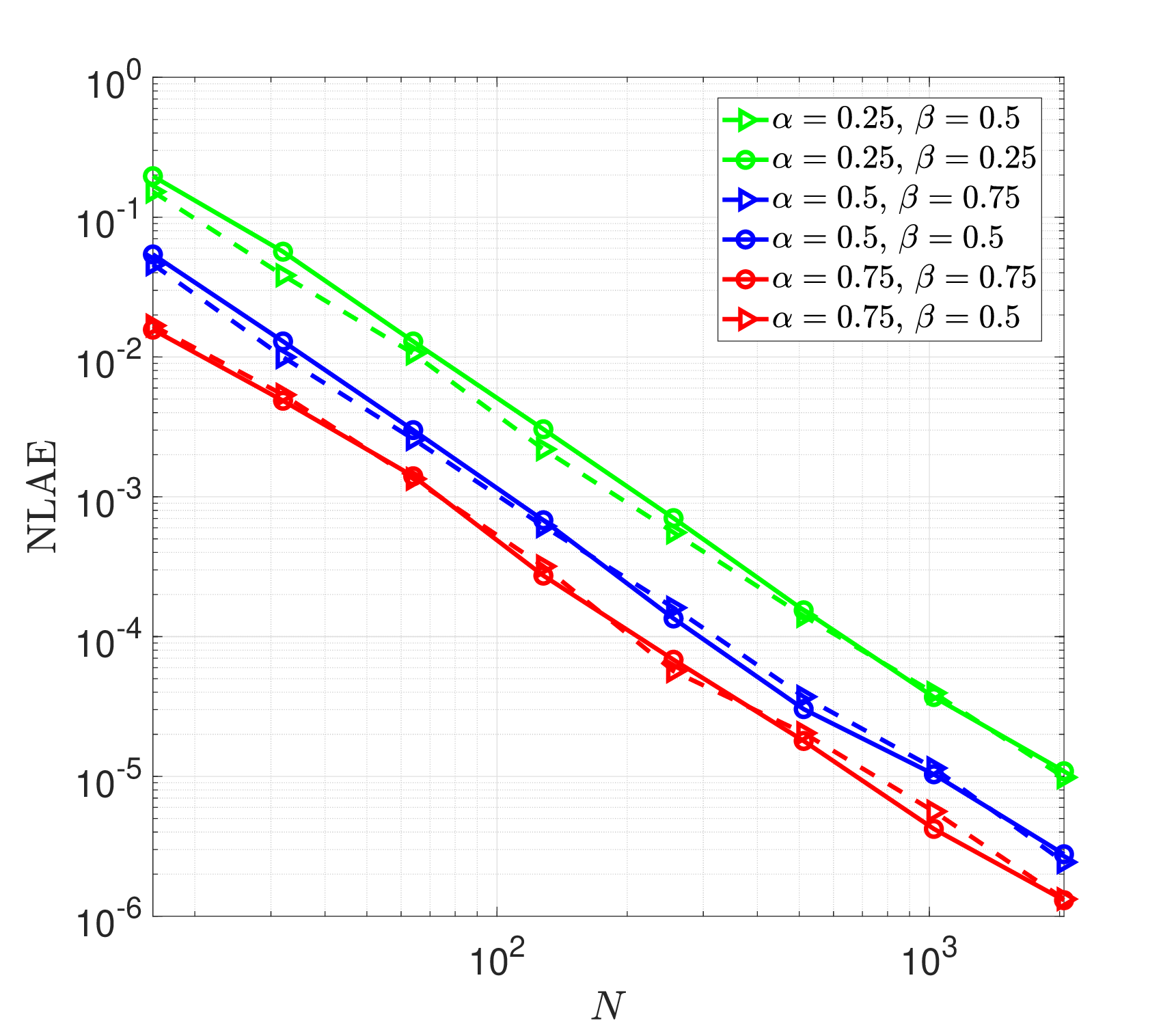}
\caption{$\mathrm{NLAE}$ versus $N$ under different values of $\alpha$ and $\beta$ for strategy (i) with $c=1/2$ and $\mathrm{SNR}=10$\,dB.}
\label{NLAE_joint_parameters}
\end{figure}

\begin{figure}[t]
\centering
\includegraphics[width=1.0\linewidth]{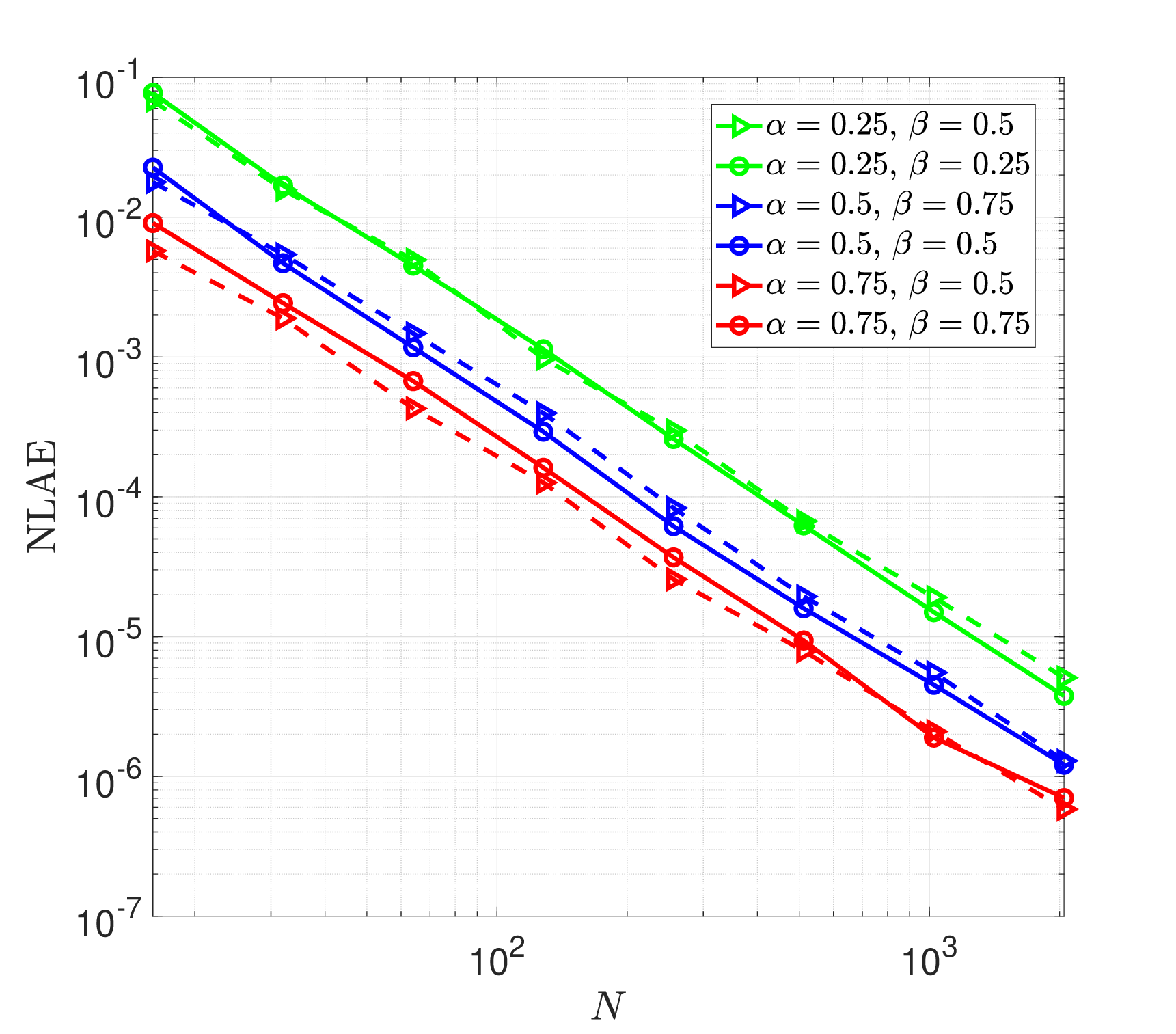}
\caption{$\mathrm{NLAE}$ versus $N$ under different values of $\alpha$ and $\beta$ for strategy (ii) with $c=1/2$ and $\mathrm{SNR}=10$\,dB.}
\label{NLAE_seq_parameters}
\end{figure}

\subsection{Empirical validation of the derived power allocation schemes}
Finally, we compare the mutual information lower bounds achieved by the two proposed estimation strategies with those obtained using the expectation-maximization (EM)-based estimator from~\cite{nayebi2017semi} and the LMMSE estimator. The EM algorithm belongs to the family of joint channel and symbol estimation strategies.  To evaluate the EM estimator, we generate $T\gg K(N-L)$ independent realizations of the received training and data signals. In each realization, the EM algorithm is applied to alternately update the channel estimate and the estimated data symbols. The estimation error is then recorded for all data symbols, and the overall error covariance matrix $\tilde{\mathbf{R}}_{\mathrm{EM}}$ is obtained by averaging these errors over the $T$ trials. Based on the error covariance, the lower bound on the achievable mutual information is computed directly from
\begin{align}\label{LB_EM}
    \frac{\mathcal{I}_{\mathrm{EM}}}{K(N-L)} =& \log_2(\pi eP) - \frac{\log_2(\det{(\pi e\mathbf{R}_{\mathrm{EM}})})}{K(N-L)}.
\end{align}
On the other hand, the LMMSE estimator belongs to the family of sequential channel and symbol estimation strategies. First, the channel is estimated based on the LMMSE technique. Then, the true channel is replaced by its LMMSE estimator and the symbol vector is estimated by applying the MMSE filter. 
By vectorizing (\ref{eq:received_sequential}), we obtain $\mathbf{y}=(\mathbf{I}_{N-L}\otimes{\hat{\mathbf{G}}})\mathbf{s}+\pmb{\omega}$, then the LMMSE estimate of the data symbols is expressed as
\begin{align}
    \hat{\mathbf{s}} =& (\mathbf{C}_{s}^{-1}+(\mathbf{I}_{N-L}\otimes \hat{\mathbf{G}})^T\mathbf{C}_{\omega}^{-1}(\mathbf{I}_{N-L}\otimes \hat{\mathbf{G}}))^{-1} \notag\\
    &\cdot(\mathbf{I}_{N-L}\otimes \hat{\mathbf{G}})^T\mathbf{C}_{\omega}^{-1}\mathbf{y},
\end{align}
where 
\begin{align}
    \mathbf{C}_s =& \mathbb{E}[\mathbf{ss}^H] = P\mathbf{I}_{K(N-L)}, \\
    \mathbf{C}_{\omega} =& \mathbb{E}[\pmb{\omega}\pmb{\omega}^H] = \Big(\frac{P}{L}\sum_{k=1}^Kd_k+\sigma_v^2\Big)\mathbf{I}_{M(N-L)}. 
\end{align}
The estimation error covariance matrix $\mathbf{R}_{\mathrm{LMMSE}}$ is then computed by averaging the symbol estimation errors across all data symbols. The corresponding lower bound on the mutual information per channel use is given by
\begin{align}\label{LB_EM}
    \frac{\mathcal{I}_{\mathrm{LMMSE}}}{K(N-L)} =& \log_2(\pi eP) - \frac{\log_2(\det{(\pi e\mathbf{R}_{\mathrm{LMMSE}})})}{K(N-L)}.
\end{align}
Next, we validate the accuracy of the proposed resource-allocation procedure by comparing the empirical throughput lower bounds  
\[
\mathcal{T}_{\rm EM}:=(1-\beta)\frac{\mathcal{I}_{\rm EM}}{K(N-L)}
\ \text{and}\ 
\mathcal{T}_{\rm LMMSE}:=(1-\beta)\frac{\mathcal{I}_{\rm LMMSE}}{K(N-L)}
\]
with their corresponding asymptotic counterparts, 
$\overline{\mathcal{T}}_{\rm joint}$ defined in~\eqref{eq:Tjoint} and 
$\overline{\mathcal{T}}_{\rm seq}$ defined in~\eqref{eq:Tseq}, respectively.

Figure~\ref{allocation_joint} depicts the lower bound on the throughput
$
\frac{(1-\beta)\mathcal{I}_{\rm EM}}{K(N-L)}
$
as a function of the training power $x=a\beta$ for various values of~$\beta$, together with the asymptotic expression $\overline{\mathcal{T}}_{\rm joint}$ (labeled "Joint").  
In this experiment, we set $N=512$, $\mathrm{SNR}=20$\,dB, $\alpha=3/4$, and $c=1/3$.  
As shown in the figure, the EM-based estimator closely follows the behavior predicted by the asymptotic expression $\overline{\mathcal{T}}_{\rm joint}$.

\begin{figure}[t]
\centering
\includegraphics[width=1.0\linewidth]{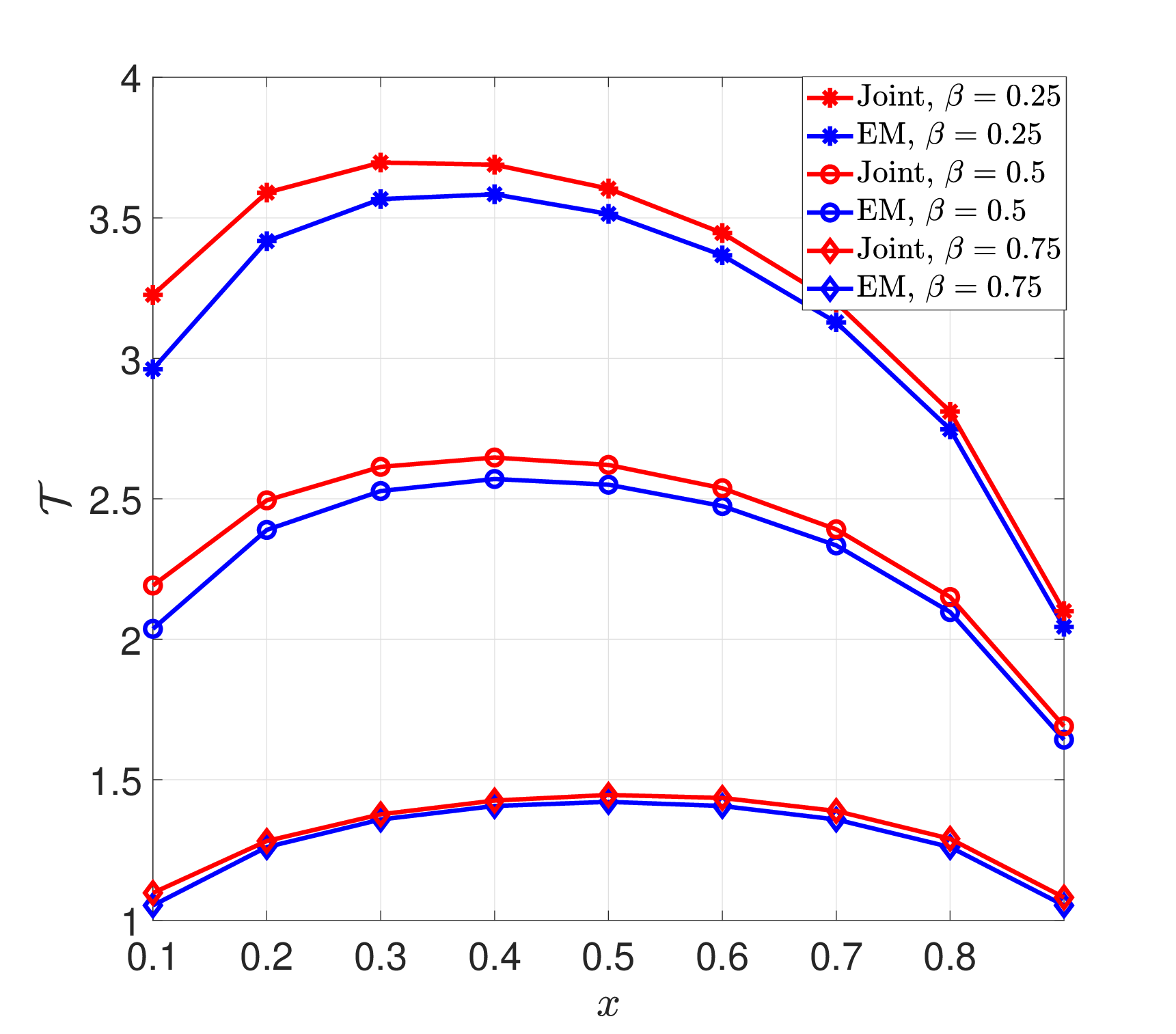}
\caption{$\mathcal{T}$ versus $x$ under $P_{\mathrm{total}}=1$ for strategy (i) with $\alpha=3/4$, $c=1/3$, and $\mathrm{SNR}=20$\,dB.}
\label{allocation_joint}
\end{figure}

It can be observed that the throughput first increases with the training power and then decreases after a certain point, which reflects the fundamental tradeoff between channel estimation accuracy and data transmission efficiency. For small values of $x$, insufficient training power leads to poor channel estimation quality, which limits the achievable throughput. As $x$ increases, allocating more power to training improves the channel estimation accuracy and thus enhances the throughput. However, when $x$ exceeds its optimal value, excessive power is devoted to training at the expense of data transmission, resulting in a throughput degradation.

We repeat the same experiment for the LMMSE symbol and channel estimator using $\mathrm{SNR}=20$\,dB, $\alpha=1$, and $c=1/2$. Figure~\ref{allocation_seq} shows that the LMMSE estimator exhibits a performance evolution that similarly matches the asymptotic prediction $\overline{\mathcal{T}}_{\rm seq}$. Compared to the EM algorithm, a noticeable discrepancy appears for small training-power values. This behavior is expected, since EM implements the maximum-likelihood estimator, whereas the LMMSE estimator used in our analysis is intrinsically suboptimal. Nevertheless, the derived expression remains highly valuable, as it still provides an accurate characterization of the overall trend and allows us to determine the optimal fraction of power that should be allocated to training.

\begin{figure}[t]
\centering
\includegraphics[width=1.0\linewidth]{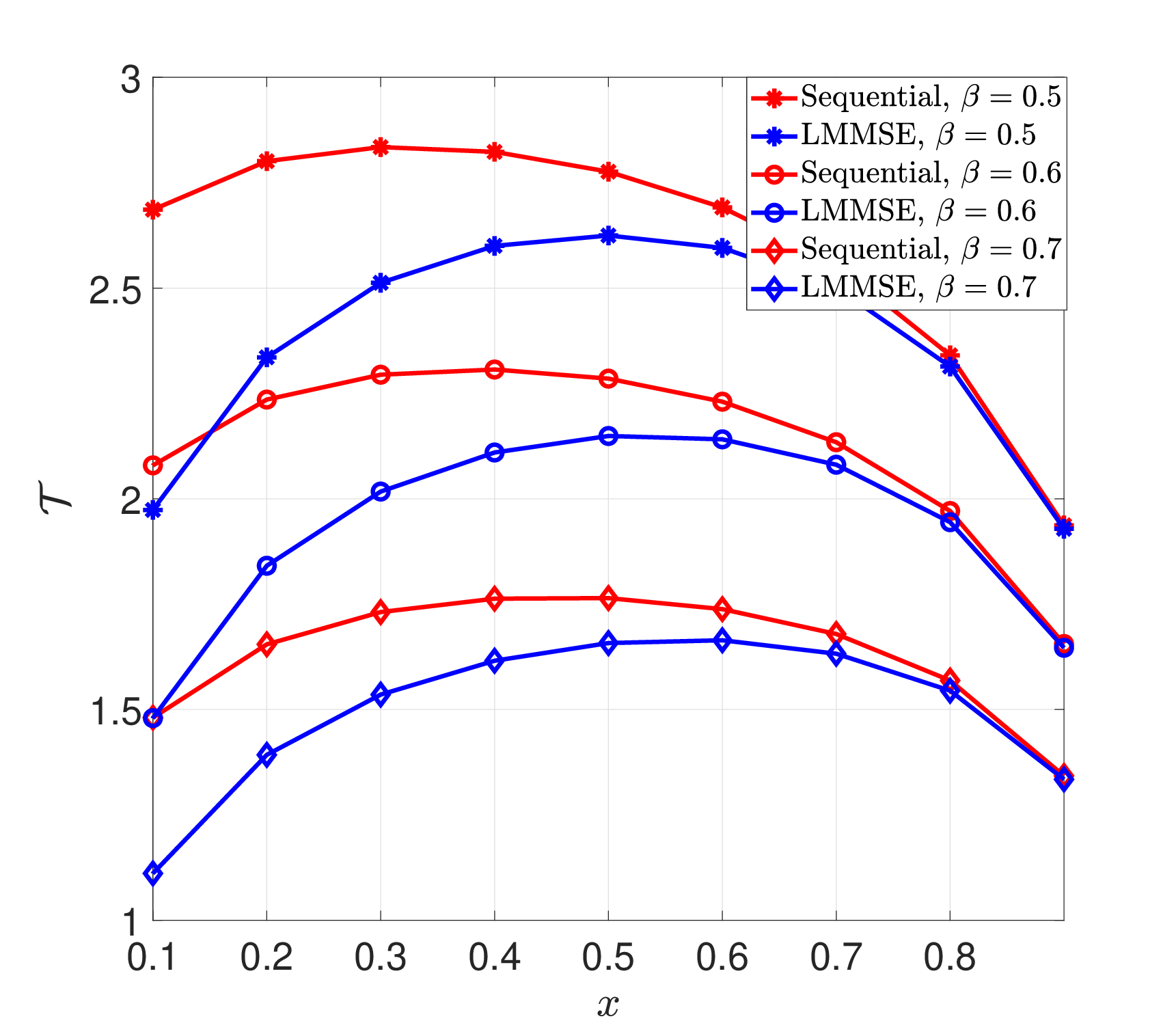}
\caption{$\mathcal{T}$ versus $x$ under $P_{\mathrm{total}}=1$ for strategy (ii) with $\alpha=1$, $c=1/2$, and $\mathrm{SNR}=20$\,dB.}
\label{allocation_seq}
\end{figure}

Overall, these results demonstrate that the proposed asymptotic framework provides not only a tight analytical approximation but also practically meaningful design guidance for training-power allocation. In particular, although the asymptotic expressions are not intended to represent exact achievable rates, they enable the direct identification of the optimal pilot-power level and yield performance improvements for both EM and LMMSE estimators in finite-dimensional systems.

\par

\section{Conclusion}
\label{sec_con}
We studied uplink channel estimation and symbol detection in multi-user systems under two different receiver strategies. In strategy (i), the BS jointly estimates the channels and data symbols, while in strategy (ii), it first obtains a channel estimate and subsequently uses it to detect the data symbols. For both strategies, we derived lower bounds on the mutual information between the transmitted data symbols and their corresponding estimates in terms of the symbol CRB matrices, and characterized their asymptotic behavior using tools from RMT. Finally, we investigated the optimal power-allocation policies that maximize the proposed mutual-information lower bounds. Beyond the analytical characterization, our results provide several useful design insights. In particular, the proposed asymptotic framework enables system designers to quantitatively assess the tradeoff between training and data transmission power. The derived optimal power-allocation policies reveal that allocating excessive power to training is inefficient and that an appropriate balance between training and data transmission is crucial to maximize the achievable throughput. Furthermore, the comparison between the joint and sequential receiver strategies offers guidance on which processing architecture should be adopted under different operating regimes.

\appendices
\section{Proof of Theorem~\ref{theorem_complex_CRB}}\label{proof_theorem_1}
Differentiating the log-likelihood function $\ln{p(\mathbf{Y},\pmb{\eta})}$ yields
\begin{align}
    &\frac{\partial\ln{p(\mathbf{Y},\pmb{\eta})}}{\partial\mathbf{s}(n)} = \frac{1}{\sigma_v^2}\mathbf{v}^H(n)\mathbf{G},\,\,\,n=L,\ldots,N-1, \notag\\ 
    &\frac{\partial\ln{p(\mathbf{Y},\pmb{\eta})}}{\partial\mathbf{s}(n)^{\ast}} = \frac{1}{\sigma_v^2}\mathbf{v}^T(n)\mathbf{G}^{\ast},\,\,\,n=L,\ldots,N-1, \notag\\
    &\frac{\partial\ln{p(\mathbf{Y},\pmb{\eta})}}{\partial\mathbf{g}_k} = \frac{1}{\sigma_v^2}\sum_{n=0}^{N-1}s_k(n)\mathbf{v}^H(n), \notag\\
    &\frac{\partial\ln{p(\mathbf{Y},\pmb{\eta})}}{\partial\mathbf{g}_k^{\ast}} = \frac{1}{\sigma_v^2}\sum_{n=0}^{N-1}s_k^{\ast}(n)\mathbf{v}^T(n). \notag
\end{align}
Then we can obtain $\pmb{\mathcal{J}}_{ss} = \mathrm{blkdiag}\{\mathbf{B},\ldots,\mathbf{B}\}$ with $\mathbf{B}=\frac{1}{\sigma_v^2}\mathbf{G}^H\mathbf{G}$, $\pmb{\mathcal{J}}_{sg} = [\mathbf{\Omega}_{L}^T,\ldots,\mathbf{\Omega}_{N-1}^T]^T$ with $\mathbf{\Omega}_n(n) = \frac{1}{\sigma_v^2}\mathbf{s}(n)^T\otimes\mathbf{G}^H$, and $\pmb{\mathcal{J}}_{gg} = \frac{1}{\sigma_v^2}\sum_{n=0}^{N-1}\mathbf{s}^{\ast}(n)\mathbf{s}(n)^T\otimes\mathbf{I}_M$. Motivated by the results in~\cite{de1997cramer,zhang2025fundamental}, the inverse of $\pmb{\mathcal{J}}_{\eta\eta}$ corresponds to the complex CRB, i.e., $\mathbf{CRB}(\mathbf{S}_d,\mathbf{G}) = \pmb{\mathcal{J}}_{\eta\eta}^{-1}$.

\par

\section{Proof of Theorem~\ref{theorem_MI_joint}}\label{app:theorem_MI_joint}
\begin{figure*}
\begin{align}\label{log_CRB}
    \frac{\log_2(\det(\pi e\mathbb{E}_{\mathbf{s}}\{\mathbf{CRB}(\mathbf{S}_d\}))}{K(N-L)} =& \log_2{\left(\pi e\sigma_v^2\right)} + \log_2{\left(1+\frac{c\alpha P}{a\beta}\right)} - \frac{1}{K}\sum_{k=1}^K\log_2(\beta_k) - \frac{1}{K}\sum\limits_{k=1}^K\log_2\left(\lambda_k(\tilde{\mathbf{H}}^H\tilde{\mathbf{H}})\right).
\end{align}
\hrulefill
\end{figure*}
From~(\ref{CRB_Symbol}), we can derive the expression as~(\ref{log_CRB}). The empirical spectral distribution (ESD) of eigenvalues of the normalized matrix ${\mathbf{H}}^H{\mathbf{H}}$ is defined as 
\begin{align}
    \mu_{\tilde{\mathbf{H}}^H\tilde{\mathbf{H}}}(\lambda) = \frac{1}{K}\sum_{k=1}^K\delta_{\lambda_k(\tilde{\mathbf{H}}^H\tilde{\mathbf{H}})},
\end{align}
where $\lambda_k(\tilde{\mathbf{H}}^H\tilde{\mathbf{H}})$ is the eigenvalues of $\tilde{\mathbf{H}}^H\tilde{\mathbf{H}}$. It is well-known from~\cite{marchenko1967distribution,tulino2004random} that as $M\rightarrow+\infty$, $K\rightarrow+\infty$ with $K/M\rightarrow c\in(0,1)$, $ \mu_{\tilde{\mathbf{H}}^H\tilde{\mathbf{H}}}(\lambda)$ converges almost surely to the non-random distribution of Marchenko-Pastur law with density function $\mu(x)$ given by
\begin{align}
\mu(x) = \frac{\sqrt{(x-a)(b-x)}}{2c\pi x}\mathbf{1}_{[a,b]}(x),
\end{align}
where $\mathbf{1}_{[a,b]}(x)$ is the indicator function of the interval $[a,b]$ with $a=(1-\sqrt{c})^2$ and $b=(1+\sqrt{c})^2$. 
Thus we can obtain,
\begin{align}\label{MP distribution}
    \frac{1}{K}\sum\limits_{k=1}^K\log_2\left(\lambda_k(\tilde{\mathbf{H}}^H\tilde{\mathbf{H}})\right) \rightarrow \int_a^b\log_2(x)\mu(x)dx.
\end{align}
which can be further simplified as \cite{bai_silverstein_book}:
\begin{align}\label{MP_int_a_b}
   &\int_a^b\log_2(x)\mu(x)dx
    = \frac{c-1}{c}\log_2{(1-c)} - \log_2(e).
\end{align}
Hence, we obtain:
\begin{align}
    &\frac{\log_2(\det(\pi e\mathbb{E}_{\mathbf{s}}[\mathbf{CRB}({\bf S}_d)])}{K(N-L)} - \log_2\left(1+\frac{c\alpha P}{a\beta}\right)\notag \\
    &\qquad\qquad - \log_2(\pi e^2\sigma_v^2) + \frac{1}{K}\sum_{k=1}^K\log_2(\beta_k) \notag\\
    &\qquad\qquad + \frac{c-1}{c}\log_2{(1-c)} \xrightarrow[]{a.s.} 0.
\end{align}
By replacing $\frac{\log_2(\det(\pi e\mathbb{E}_{\mathbf{s}}[\mathbf{CRB}({\bf S}_d)])}{K(N-L)}$ by its almost sure limit, we get the asymptotic expression for $\mathcal{I}_{\rm joint}$.

\par

\section{Proof of Theorem \ref{sequential_bcrb}}
\label{app:sequential_bcrb}
Let ${\bf y}=[{\bf y}^T(L),\cdots,{\bf y}^T(N-1)]^{T}$ and $\boldsymbol{\omega}=[\boldsymbol{\omega}(L),\cdots,\boldsymbol{\omega}(N-1)]^{T}$. Then, ${\bf y}$ writes as:
$$
{\bf y}=({\bf I}_{N-L}\otimes \hat{\bf G}){\bf s}+\boldsymbol{\omega}.
$$
Then, the log likelihood function writes as:
\begin{align}
&\ln p({\bf y}, {\bf s} | \hat{\bf G})={\rm const} -P^{-1}{\bf s}^{H}{\bf s} -\ln {\rm det}({\bf R}) \nonumber \\
&\quad-({\bf y}-(I_{N-L}\otimes\hat{\bf G}){\bf s})^{H}{\bf R}^{-1}({\bf y}-(I_{N-L}\otimes\hat{\bf G}){\bf s}),
\end{align}
where $\mathbf{R}$ is expressed as (\ref{R_expression}) with ${\bf A}=\frac{1}{L}{\bf S}_d^{T}{\bf D}{\bf S}_d^\ast+\sigma_v^{2}{\bf I}_{N-L}$ and ${\bf D}=\mathrm{diag}\left\{\left[\frac{\beta_1\sigma_v^2}{a\beta_1+\frac{M}{L}\sigma_v^2},\ldots,\frac{\beta_K\sigma_v^2}{a\beta_K+\frac{M}{L}\sigma_v^2}\right]\right\}$. The FIM for the complex variable ${\bf s}$ is given by:  
\begin{figure*}
\begin{align}\label{R_expression}
{\bf R}=& \left(\begin{bmatrix}
\frac{1}{L}{\bf s}(L)^{H}{\bf D}{\bf s}(L) & \cdots & \frac{1}{L}{\bf s}(N-1)^{H}{\bf D}{\bf s}(L) \\
\vdots &\ddots &\vdots \\
\frac{1}{L}{\bf s}(L)^{H}{\bf D}{\bf s}(N-1) & \cdots & \frac{1}{L}{\bf s}(N-1)^{H}{\bf D}{\bf s}(N-1) 
\end{bmatrix}+\sigma_v^2{\bf I}_{N-L}\right)\otimes\mathbf{I}_M = \mathbf{A} \otimes \mathbf{I}_M.
\end{align}
\hrulefill
\end{figure*}
$$
{\bf F}=\mathbb{E}_{{\bf y},{\bf s}| \hat{\bf G}} \left[\frac{\partial \ln p({\bf y},{\bf s}|\hat{\bf G})}{\partial {\bf s}^\ast}\left(\frac{\partial \ln p({\bf y},{\bf s}|\hat{\bf G})}{\partial {\bf s}^\ast}\right)^{H}\right]. 
$$

Let ${\bf W}=[\boldsymbol{\omega}(L),\cdots,\boldsymbol{\omega}(N-L)]$, using the following relations:
\begin{align}
&\frac{\partial -\ln {\rm det}({\bf R})}{\partial {\bf s}^\ast}=-\frac{M}{L}({\bf A}^{-1}\otimes {\bf D}){\bf s},\notag\\
&\frac{\partial -({\bf y}-{\bf I}_{N-L}\otimes \hat{\bf G}{\bf s})^{H}{\bf R}^{-1}({\bf y}-{\bf I}_{N-L}\otimes \hat{\bf G}{\bf s})}{\partial {\bf s}^\ast} \notag\\
&\quad=({\bf A}^{-1}\otimes \hat{\bf G}^{H})\boldsymbol{\omega}+\frac{1}{L}({\bf A}^{-1}\otimes {\bf DS}_d{\bf A}^{-T}){\rm vec}({\bf W}^{H}{\bf W}), \notag\\
&\frac{\partial -P^{-1}{\bf s}{\bf s}^{H}}{\partial {\bf s}^\ast}=-P^{-1}{\bf s}. \notag
\end{align}
Hence, using Lemma \ref{lem:technical} in Appendix \ref{app:technical}, we obtain:
\begin{align}
{\bf F}=&\mathbb{E}_{{\bf s}}\left[\frac{M^2}{L^2}({\bf A}^{-1}\otimes {\bf D}){\bf s}{\bf s}^{H}({\bf A}^{-1}\otimes {\bf D})+({\bf A}^{-1}\otimes \hat{\bf G}^{H}\hat{\bf G})\right]\nonumber\\
&-\mathbb{E}_{{\bf s}}\left[\frac{2M}{L^2}({\bf A}^{-1}\otimes {\bf D}){\bf s}{\bf s}^{H}({\bf A}^{-1}\otimes {\bf D})\right]\nonumber \\
&+\mathbb{E}_{{\bf s }}\left[\frac{1}{L^2}({\bf A}^{-1}\otimes {\bf DS}_d{\bf A}^{-T})M^2{\rm vec}({\bf A}^{T}){\rm vec}({\bf A}^T)^{H}\right. \nonumber \\
&\left.\quad\times ({\bf A}^{-1}\otimes {\bf A}^{-T}{\bf S}_d^{H}{\bf D})\right] +\mathbb{E}_{{\bf s}}\left[\frac{1}{L^2}({\bf A}^{-1}\otimes {\bf DS}_d{\bf A}^{-T})\right. \nonumber\\
&\left.\quad\times(M{\bf A}\otimes {\bf A}^{T})({\bf A}^{-1}\otimes {\bf A}^{-T}{\bf S}_d^{H}{\bf D})\right] + P^{-2}\mathbb{E}_{{\bf s}}[{\bf s}{\bf s}^{H}]\nonumber\\
=&\mathbb{E}_{{\bf s}}[({\bf A}^{-1}\otimes \hat{\bf G}^{H}\hat{\bf G})]+ P^{-2}\mathbb{E}_{{\bf s}}[{\bf s}{\bf s}^{H}]\notag\\
&+\frac{M}{L^2}\mathbb{E}_{{\bf s}}\left[({\bf A}^{-1}\otimes {\bf DS}_d{\bf A}^{-T}{\bf S}_d^{H}{\bf D})\right] \notag\\
=&\mathbb{E}_{{\bf s}}\left[{\bf A}^{-1}\otimes \left(\hat{\bf G}^{H}\hat{\bf G}+\frac{M}{L^2}{\bf DS}_d{\bf A}^{-T}{\bf S}_d^{H}{\bf D}\right)\right]+\frac{{\bf I}_{K(N-L)}}{P}. \notag
\end{align}


\section{Proof of Theorem \ref{theorem_MI_sequential}}
\label{app:theorem_MI_sequential}
To begin with, using Woodbury matrix identity, we note that:
\begin{align}
    {\bf D}\frac{1}{L}({\bf S}_d{\bf A}^{-T}{\bf S}_d^{H}){\bf D}={\bf D}-\sigma_v^2{\bf D}^{\frac{1}{2}}{\bf Q}{\bf D}^{\frac{1}{2}},
\end{align}
where ${\bf Q}=\big(\frac{1}{L}{\bf D}^{\frac{1}{2}}{\bf S}_d{\bf S}_d^{H}{\bf D}^{\frac{1}{2}}+\sigma_v^2{\bf I}_K\big)^{-1}$. We thus obtain:
\begin{align}
&{\bf BCRB}^{-1}\notag\\
&=\mathbb{E}_{{\bf s}}\left[{\bf A}^{-1}\otimes \Big(\hat{\bf G}^{H}\hat{\bf G}+\frac{M}{L}{\bf D}-\frac{M\sigma_v^2}{L}{\bf D}^{\frac{1}{2}}{\bf Q}{\bf D}^{\frac{1}{2}}\Big)\right]\nonumber \\
&\quad+P^{-1}{\bf I}_{K(N-L)}. 
\end{align}
Let ${\bf S}_d^{T}={\bf U}_{S}\boldsymbol{\Lambda}_{S}{\bf V}_S^{H}$ be a singular value decomposition of ${\bf S}^{T}$ where ${\bf U}_{S}$ and ${\bf V}_S$ are unitary matrices representing the left and right singular vectors of ${\bf S}_d^{T}$. Since ${\bf S}_d$ is a Gaussian matrix with independent and identically distributed elements, ${\bf U}_S$ and ${\bf V}_S$ are Haar distributed and are independent. Hence, 
$$
\mathbb{E}_{{\bf U}_{S}}[{\bf A}^{-1}]=\mathbb{E}_{{\bf V}_S,\boldsymbol{\Lambda}_S}\{\theta_1 {\bf I}_{N-L}\},
$$
where
$$
\theta_1=\frac{1}{N-L}{\rm tr}\left(\frac{1}{L}\boldsymbol{\Lambda}_S{\bf V}_S^{H}{\bf D}{\bf V}_S\boldsymbol{\Lambda}_S+\sigma_v^2{\bf I}_{N-L}\right)^{-1}. 
$$
By noting that ${\bf Q}$ is independent of ${\bf U}_S$, the inverse of the CRB simplifies to:
\begin{align}
&{\bf BCRB}^{-1}\notag\\
&=\mathbb{E}\left[\theta_1{\bf I}_{N-L}\otimes \Big(\hat{\bf G}^{H}\hat{\bf G}+\frac{M}{L}{\bf D}-\frac{M\sigma_v^2}{L}{\bf D}^{\frac{1}{2}}{\bf Q}{\bf D}^{\frac{1}{2}}\Big)\right]\nonumber \\
&\quad+P^{-1}{\bf I}_{K(N-L)},
\end{align}
where the expectation is taken over the distribution of $\boldsymbol{\Lambda}_S$ and ${\bf V}_S$.
From the standard results of random matrix theory~\cite{hachem2013bilinear}, we know that $\theta_1$ converges to 
\begin{equation}
\theta_1-\frac{1}{\sigma_v^2(1+\overline{m})}\xrightarrow[]{a.s} 0, \label{eq:conv_result_1}
\end{equation}
where $\overline{m}$ is the unique positive solution in $m$ to the following equation:
$$
m=\frac{1}{N}{\rm tr}\left(\tilde{\bf D}\Big({\frac{\sigma_v^2}{P}{\bf I}_K+\frac{1-\beta}{(1+m)}}\tilde{\bf D}\Big)^{-1}\right),
$$
with $\tilde{\bf D}$ being defined in \eqref{eq:tilde_D}
and that:
\begin{equation}
\|\mathbb{E}_{{\bf S}}[{\bf Q}]-{\bf T}\|_2\xrightarrow[]{a.s.}0, \label{eq:conv_result_2}
\end{equation}
where ${\bf T}=(\sigma_v^{2}{\bf I}_K+\frac{(1-\beta)P}{(1+\overline{m})}\tilde{\bf D})^{-1}$. 

Using the convergence results in \eqref{eq:conv_result_1} and \eqref{eq:conv_result_2}, we thus obtain (\ref{asy_det_BCRB}). Furthermore, based on (\ref{covariance_hat_g}), let ${\bf R}_{\hat{G}}$ be the $K\times K $ diagonal matrix given by:
\begin{figure*}
\begin{align}\label{asy_det_BCRB}
-\frac{1}{K(N-L)}\log_2(\det({\bf BCRB})) + \log_2\big(\sigma_v^2(1+\overline{m})\big)-\frac{1}{K}\,
 \log_2\Big(\det\Big(
   \hat{\bf G}^{H}\hat{\bf G}
   + \frac{\alpha(1-\beta)P}{(1+\overline{m})}\,\tilde{\bf D}^2{\bf T}
   + \frac{\sigma_v^2(1+\overline{m})}{P}\,{\bf I}_K\Big)
 \Big) \xrightarrow[]{a.s.}0.
\end{align}
\hrulefill
\end{figure*}
$$
{\bf R}_{\hat{G}}= \mathbb{E}[\hat{\mathbf{G}}^H\hat{\mathbf{G}}] =a{\bf B}^2\Big(a{\bf B}+\sigma_v^2\frac{\alpha}{\beta}{\bf I}_K\Big)^{-1}\!\!=\!\!{\bf B}^2\Big({\bf B}+\sigma_v^2\frac{\alpha}{x}{\bf I}_K\Big)^{-1}\!\!. 
$$
Using \cite[Lemma 2]{KammounCND13}, we obtain:
\begin{align}
    \frac{1}{K}\log_2\Big(\det\Big(
   \hat{\bf G}^{H}\hat{\bf G}
   + \frac{\alpha(1-\beta)P}{(1+\overline{m})}\,\tilde{\bf D}^2{\bf T}
   + \frac{\sigma_v^2(1+\overline{m})}{P}\,{\bf I}_K
 \Big)\Big)   \notag\\
 -\frac{1}{K}\overline{V}\xrightarrow[]{a.s.}0, \notag
\end{align}
where
\begin{align}
&\overline{V}:=\log_2\Big({\rm det}\Big(\frac{{\bf R}_{\hat{G}}}{1+\kappa}+\frac{\alpha(1-\beta)P}{(1+\overline{m})}\,\tilde{\bf D}^2{\bf T}\nonumber\\
&+\frac{\sigma_v^2(1+\overline{m})}{P}{\bf I}_K\Big)\Big)+M\log_2(1+\kappa)-\frac{M\kappa}{(1+\kappa)\ln(2)}, \notag
\end{align}
with $\kappa$ being the unique solution to the following equation:
$$
\kappa=\frac{1}{M}{\rm tr}\Big({\bf R}_{\hat{G}}\Big(\frac{{\bf R}_{\hat{G}}}{1+\kappa}+\frac{\alpha(1-\beta)P}{(1+\overline{m})}\,\tilde{\bf D}^2{\bf T}
   +\frac{\sigma_v^2(1+\overline{m})}{P}{\bf I}_K\Big)^{-1}\Big).
$$

\par

\section{Technical lemma}
\label{app:technical}
\begin{lemma}
\label{lem:technical}
Let ${\bf W}$ be an $M_1\times M_2$ complex circularly symmetric Gaussian matrix with zero mean entries satisfying:
$$
\mathbb{E}[W_{kn}W_{jm}^{\ast}]=A_{nm} \delta_{kj}, 
$$
where $\{A_{nm}\}_{n,m=1}^{M_2}$ are some real entries. 
Let ${\bf A}$ be the $M_2\times M_2$ matrix such that $[{\bf A}]_{n,m}=A_{nm}$. Then,
\begin{align}
&\mathbb{E}[{\rm vec}({\bf W}^{H}{\bf W}){\rm vec}^{H}({\bf W}^{H}{\bf W})] \notag\\
&\qquad=M_1^2{\rm vec}({\bf A}^{T}){\rm vec}({\bf A})^{T}+M_1{\bf A}\otimes {\bf A}^{T}.
\end{align}
\end{lemma}
\begin{proof}
Using the identity ${\rm vec}({\bf B})=\sum_{i,j} B_{ij}{\rm vec}({\bf E}_{ij})$ where ${\bf E}_{ij}={\bf e}_i{\bf e}_j^{T}$ where ${\bf e}_i$ and ${\bf e}_j$ are the canonical vectors of the canonical basis, we obtain:
\begin{align}
&\mathbb{E}[{\rm vec}({\bf W}^{H}{\bf W}){\rm vec}^{H}({\bf W}^{H}{\bf W})]\nonumber\\
&=\sum_{i,k,j,\ell=1}^{M_2}\sum_{n,m=1}^{M_1} \mathbb{E}[W_{ni}^\ast W_{nj} W_{mk}W_{m\ell}^\ast]{\rm vec}({\bf E}_{ij}){\rm vec}^{T}({\bf E}_{k\ell})\nonumber\\
&=\sum_{i,k,j,\ell=1}^{M_2}\sum_{n,m=1}^{M_1} \mathbb{E}[W_{ni}^\ast W_{nj}]\mathbb{E}[W_{mk}W_{m\ell}^\ast]{\rm vec}({\bf E}_{ij}){\rm vec}^{T}({\bf E}_{k\ell})\nonumber \\
&+\sum_{i,k,j,\ell=1}^{M_2}\sum_{n,m=1}^{M_1} \mathbb{E}[W_{ni}^\ast W_{mk}] \mathbb{E}[W_{nj}W_{m\ell}^\ast]{\rm vec}({\bf E}_{ij}){\rm vec}^{T}({\bf E}_{k\ell}) \nonumber\\
&=M_1^2 \sum_{i,j=1}^{M_2}\sum_{k,\ell=1}^{M_2}A_{ji}A_{k\ell} ({\bf e}_j\otimes {\bf e}_i)({\bf e}_\ell^{T}\otimes {\bf e}_k^{T}) \nonumber \\
&\quad+M_1\sum_{i,k=1}^{M_2}\sum_{j,\ell=1}^{M_2} A_{ki}A_{j\ell}({\bf e}_j\otimes {\bf e}_i)({\bf e}_\ell^{T}\otimes {\bf e}_k^{T}) \nonumber\\
&=M_1^2{\rm vec}({\bf A}^{T}){\rm vec}({\bf A})^{T}+M_1{\bf A}\otimes {\bf A}^{T}. \notag
\end{align}
\end{proof}

\bibliographystyle{IEEEtran}
\bibliography{references}

\end{document}